\documentclass[12pt]{article}\usepackage{amsmath, amsthm, alltt, amssymb, xspace, times, algorithm, algorithmic}
\usepackage[T1]{fontenc}
\usepackage[utf8]{inputenc}
\usepackage[english]{babel}
\usepackage{csquotes}
\usepackage{graphicx}
\usepackage{xcolor}
\usepackage{comment}
\usepackage{hyperref}
\usepackage{stmaryrd}
\usepackage{tikz}
\usetikzlibrary{decorations.pathmorphing}
\usetikzlibrary{arrows.meta}
\usepackage{mathscinet}
\usepackage{footnote}
\usepackage{chngcntr}
\makesavenoteenv{tabular}
\usepackage{subfig}
\usepackage{enumitem}

\numberwithin{equation}{section}
\numberwithin{table}{section}
\newtheorem{theorem}{Theorem}[section]

\newtheorem{corollary}[theorem]{Corollary}
\newtheorem{lemma}[theorem]{Lemma}
\newtheorem{proposition}[theorem]{Proposition}

\theoremstyle{definition}
\newtheorem{definition}[theorem]{Definition}
\newtheorem{remark}[theorem]{Remark}
\newtheorem{example}[theorem]{Example}

\newcommand{\recht}{\operatorname}

\newcommand{\dg}[0]{\digamma}

\DeclareRobustCommand{\VAN}[3]{#2}

\allowdisplaybreaks

\title{Information-theoretic metrics for Local Differential Privacy protocols}
\author{Milan Lopuhaä-Zwakenberg, Boris \v{S}kori\'{c}, and Ninghui Li}
\date{\today}

\begin{document}

\maketitle

\subsection*{Abstract}
Local Differential Privacy (LDP) protocols allow an aggregator to obtain population statistics about sensitive data of a userbase, while protecting the privacy of the individual users. To understand the tradeoff between aggregator utility and user privacy, we introduce new information-theoretic metrics for utility and privacy. 
Contrary to other LDP metrics, these metrics highlight the fact that the users and the aggregator are interested in fundamentally different domains of information. We show how our metrics relate to $\varepsilon$-LDP, the \emph{de facto} standard privacy metric, giving an information-theoretic interpretation to the latter.
Furthermore, we use our metrics to quantitatively study the privacy-utility tradeoff for a number of popular protocols.

\textbf{key words:} local differential privacy, privacy metrics, utility metrics, information theory, privacy-utility tradeoff

\section{Introduction}
\subsection{Privacy and utility metrics in Local Differential Privacy}

In a context where a data aggregator collects potentially sensitive data, 
there is an inherent tension between the aggregator's desire to obtain accurate population statistics 
and the individuals' desire to protect their private data. 
One approach to protect privacy is offered by \emph{Local Differential Privacy} (LDP) protocols \cite{kasiviswanathan2011can}. 
In this approach, each user randomises his private data before sending it to the aggregator. 
This hides the users' true data, while for a large population size the randomness of the users cancels out, 
which allows the aggregator to obtain an accurate estimate of the population statistics. 
LDP-mechanisms are widely used in industry by companies such as Apple \cite{apple2017learning}, Google \cite{erlingsson2014rappor}, and Microsoft \cite{ding2017collecting}.

One of the main settings in which LDP protocols are used is that of \emph{frequency estimation} \cite{warner1965randomized,erlingsson2014rappor,wang2017locally}. 
In this setting, 
every user has a private data item from a (finite) set $\mathcal{A}$, 
and the aggregator's goal is to determine the frequencies of the elements of $\mathcal{A}$ among the user population. 
In this paper, we focus on this setting, as it is both well studied in the literature and frequently applied in practice.

The tradeoff between privacy and utility also manifests itself in the context of LDP. 
Intuitively, the more `random' the privacy protocol is, the better it will hide an individual's private data, 
but the more noisy the aggregator's estimations will be. 
In order to characterise this tradeoff, there is a need for metrics that measure the privacy and utility of a given protocol. 
Using these metrics, one can define a minimal level of privacy that a protocol has to provide, 
and maximise utility given this restriction. 

The \emph{de facto} standard metric for the privacy of an LDP protocol is \emph{$\varepsilon$-LDP} \cite{kasiviswanathan2011can}, 
which is an adaptation of $\varepsilon$-differential privacy, 
a metric used in the context where the aggregator is trusted and releases randomised versions of (queries on) a hidden database. Informally, a low value of $\varepsilon$ means that the probability distribution of the output of the protocol does not depend too much on the input. In this way, the aggregator cannot determine the user's hidden data from their public output.

$\varepsilon$-LDP gives an upper bound on privacy leakage in a worst-case sense. The strength of this approach is that it gives strong privacy guarantees that hold in any situation. However, in practice it can be difficult to fulfill such a stringent definition of privacy \cite{korolova2009releasing,haeberlen2011differential,gotz2011publishing}, and for this reason many relaxations of $\varepsilon$-(L)DP have been introduced \cite{dwork2006our,cuff2016differential,dwork2016concentrated,mironov2017renyi}. Also, the worst-case privacy leakage can in practice be quite far from the typical leakage of the average user. Therefore, it is important to study the privacy of protocols not only in terms of its worst-case \emph{guarantees}, but also in terms of its typical \emph{privacy behaviour}.

A natural `language' for studying the behaviour of privacy protocols is information theory, and several information-theoretic privacy (or leakage) metrics have been proposed, both in the central and local settings of differential privacy \cite{mcgregor2010limits,barthe2011information,alvim2011differential,mir2012information,du2012privacy,duchi2013local,cuff2016differential,wang2016relation}. Many of these define leakage in terms of the mutual information between a user's data, and the obfuscated output they send to the aggregator. 
However, the aggregator already gains information about a user's data via the population statistics inferred from other users' data; hence the mutual information between input and output does not accurately represent the information leaked by the protocol if one does not condition for the population statistics \cite{li2009tradeoff,cuff2016differential}.

On the side of utility, the typical way of measuring the performance of a protocol is to give a frequency oracle, i.e. an algorithm that, given the noisy data from the user population, estimates the frequencies of the data items, and then measure the average error of this estimation \cite{wang2017locally,blasiok2019towards}. 
This method has the downside that even for relatively simple protocols, one has many different frequency oracles, and it is not a priori clear which one one should adopt \cite{wang2019consistent}. As such, this approach does not measure the utility of the protocol itself, but of the combination of protocol and frequency oracle. Furthermore, the average error will depend on the input distribution, and is hard to study theoretically \cite{blasiok2019towards}. This makes it hard to determine optimal protocols if information such as the number of users, or characteristics of the input distribution, are unknown a priori. 

There have been multiple approaches towards defining an information-theoretic utility measure in the LDP setting \cite{kairouz2014extremal,duchi2013local}. Some of these take the mutual information between input and output as a utility measure \cite{kairouz2014extremal,wang2019consistent}, as others would define leakage (see above). However, the aggregator is interested only in the statistics of the overall population, rather than each user's individual value \cite{li2009tradeoff}. As such, incorporating all information the aggregator receives into the utility metric leads to an overstatement of the aggregator's means to find the relevant information. 

\subsection{Contributions and outline}

The main contribution of this paper is the introduction of new information-theoretic metrics for privacy and utility in the context of LDP protocols. 
Let $\mathcal{A}$ be a finite set, of which every user $i\in\{1,\ldots,n\}$ 
has an element $X_i \in \mathcal{A}$ as its private data; we write $\vec{X} := (X_1,\cdots,X_n)$. We model the $X_i$ as being drawn independently according to an unknown probability distribution $P = (P_{x})_{x \in \mathcal{A}}$ on $\mathcal{A}$. 
The goal of the aggregator is either to learn $P$ or 
the tally vector $T = (T_x)_{x \in \mathcal{A}}$, where $T_x = \#\{i: X_i = x\}$. 
Since $P$ is unknown to the aggregator, we model it as being itself a random variable,
with prior distribution $\mu$ which reflects the aggregator's prior knowledge.

Although the total information contained in the users' data is $\recht{H}(\vec{X})$, 
we will argue that the part relevant to the aggregator is only 
$\recht{I}(\vec{X};P)$ or $\recht{H}(T)$. 
At the same time, we will argue that the privacy-sensitive information is only $\recht{H}(\vec{X}|P)$. 
To protect his private data, user $i$ computes an obfuscated data item $Y_i = Q(X_i)$, 
obtained via a (typically nondeterministic) function $Q\colon \mathcal{A} \rightarrow \mathcal{B}$ for a second finite set $\mathcal{B}$, 
and sends $Y_i$ to the aggregator; 
the aggregator only has the noisy data vector $\vec{Y}$ at his disposal. 
This leads to the following definitions of \emph{distribution utility}, \emph{tally utility} and \emph{average privacy} (Defs. \ref{def:priv} and \ref{def:uti}):
\begin{equation}
\recht{U}^{\recht{distr}}_{n,\mu}(Q) = \frac{\recht{I}(\vec{Y};P)}{\recht{I}(\vec{X};P)}, \ \ \ 
\recht{U}^{\recht{tally}}_{n,\mu}(Q) = \frac{\recht{I}(\vec{Y};T)}{\recht{H}(T)}, \ \ \ 
\recht{S}_{\mu}(Q) = \frac{\recht{H}(\vec{X}|\vec{Y},P)}{\recht{H}(\vec{X}|P)}.
\end{equation}
Note that the dependence on $\mu$ is implicit, as it is the distribution of the random variable $P$. Putting these equations into words, utility is the fraction of relevant information that the aggregator has access to, while privacy is the fraction of sensitive information that is withheld from the aggregator. We will argue that these metrics have several advantages over existing metrics: they offer an intuitive meaning of privacy and utility, they are applicable to a wide range of privacy protocols, and they take into account the fact that utility deals with a domain of information that is fundamentally different from the domain of information that is covered by privacy.

We discuss the mathematical properties of these new metrics. There are four main results:

\begin{itemize}[wide=0pt]
\item \emph{Worst case privacy} (Section \ref{sec:wc}):
We introduce a natural worst-case counterpart to $\recht{S}_{\mu}$, and we show that it is equivalent to $\varepsilon$-LDP (Theorem \ref{thm:wc}). This shows that $\varepsilon$-LDP has a natural information-theoretic interpretation.
\item \emph{Limit behaviour} (Section \ref{sec:limit}):
We study the limit behaviour of the metrics $\recht{U}^{\recht{distr}}_{n,\mu}$ and $\recht{U}^{\recht{tally}}_{n,\mu}$. This limit behaviour can be expressed in terms of an asymptotic utility $\recht{U}^{\recht{as}}_{\mu}$, yielding a utility metric only dependent on $Q$ and $\mu$. We give an explicit formula to calculate $\recht{U}^{\recht{as}}_{\mu}$.
\item \emph{Combining protocols} (Section \ref{sec:combination}):
We discuss the behaviour of privacy and utility metrics under three commonly used methods of protocol combination: sequential, parallel and mixture. 
In particular, we show that the concept of a `privacy budget' \cite{dwork2006calibrating} also applies to $\recht S_\mu$.
\item \emph{Privacy-versus-utility tradeoff} (Section \ref{sec:tradeoff}): 
We derive an upper bound on the asymptotic utility in terms of the worst-case privacy. 
This bound does not depend on the prior distribution $\mu$.
\end{itemize}
Furthermore we study two additional aspects of the metrics:
\begin{itemize}[wide=0pt]
\item In section \ref{sec:posterior} we discuss the relation of the new utility metric with the posterior distribution of $P$, 
and how the aggregator can compute the posterior. 
\item In section \ref{sec:protocols} we apply the new metrics to two established protocols, 
Generalised Randomised Response \cite{warner1965randomized} and Unary Encoding \cite{wang2017locally}.
Although the new metrics can be computationally expensive, we show that in concrete cases the computational complexity can be significantly reduced as compared to the general case.
\end{itemize}
For readability, we present the proofs of all mathematical statements in the appendices.

\section{Preliminaries} \label{sec:preliminaries}

\subsection{Notation}

Where possible, random variables are denoted with capitals, their realisations with lowercase letters, and sets with calligraphic capitals.
The identity operation on $\cal A$ is written as ${\rm id}_{\cal A}$.
If $\mathcal{A}$ is finite, then we denote the set of all probability distributions on $\mathcal{A}$ by $\mathcal{P}_{\mathcal{A}}$; 
for $\mathcal{A} = \{1,\ldots,a\}$ the set $\mathcal{P}_{\mathcal{A}}$ consists of all $p = (p_1,\cdots,p_a) \in \mathbb{R}^a$ such that $p_x \geq 0$ for all $x$ and $\sum_{x=1}^a p_x = 1$. 
Expectation over a random variable $X$ is written as $\mathbb E_X$. Throughout we use `$\recht{log}$' to denote the natural logarithm, and we use the natural logarithm as the basis for our information-theoretic definitions (see section \ref{ssec:it}). 
Taking another base does not change anything substantial. 
We write $\delta(\cdot)$ for the Dirac delta function.

\subsection{Privacy protocols} \label{ssec:priv}

Let $\mathcal{A}$ be a finite set. A \emph{privacy protocol} for $\mathcal{A}$ is a pair $Q = (\tilde{Q},\mathcal{B})$ where $\mathcal{B}$ is a finite set, and $\tilde{Q} = (\tilde{Q}_{y|x})_{y \in \mathcal{B}, x \in \mathcal{A}} \in \mathbb{R}_{\geq 0}^{\mathcal{B} \times \mathcal{A}}$ is a collection of nonnegative reals such that $\sum_y \tilde{Q}_{y|x} = 1$ for all $x$. We will often identify $\mathcal{A} = \{1,\cdots,a\}$ and $\mathcal{B} = \{1,\cdots,b\}$, and consider $\tilde{Q}$ as a $(b \times a)$-matrix. We can consider $Q$ as a random function: for $x \in \mathcal{A}$, we let $Q(x) \in \mathcal{B}$ be the random variable with probability distribution $\tilde{Q}_{\bullet|x}$.
Furthermore, we let $Q_*\colon \mathcal{P}_{\mathcal{A}} \rightarrow \mathcal{P}_{\mathcal{B}}$ be the map that sends $p$ to $\tilde{Q} \cdot p$, where $\cdot$ denotes matrix-vector multiplication. Note that if $X \sim p$, then $Q(X) \sim Q_*p$.

\subsection{Information-theoretic concepts} 
\label{ssec:it}

Let $X \in \mathcal{A}$ be a discrete random variable. The \emph{Shannon entropy} of $X$ is defined as $\recht{H}(X) = \sum_{x \in \mathcal{A}} \recht{p}_x \recht{log}\frac{1}{\recht{p}_x}$;\footnote{With the understanding that $0 \cdot \recht{log}\frac{1}{0} = 0$.} 
this measures the uncertainty about $X$. If $Y \in \mathcal{B}$ is another random variable, then the \emph{conditional entropy} of $X$ given $Y$ is defined as $\recht{H}(X|Y) = \mathbb{E}_{y}[\recht{H}(X|Y = y)]$; 
this measures the uncertainty about $X$, provided that one knows $Y$. 
Note that $Y$ does not need to be discrete. The complement of conditional entropy is the \emph{mutual information} $\recht{I}(X;Y) := \recht{H}(X) - \recht{H}(X|Y)$, which is the amount of information learned from $X$ by knowing $Y$, or vice versa (see below). 
One can write this as
$\recht{I}(X;Y) = \mathbb{E}_{x,y}\recht{log}\frac{\recht{p}_{x|Y = y}}{\recht{p}_x}$, and the result will always be nonnegative.

Let $P$ be a smooth random variable on $\cal A$, i.e. it has a smooth probability density function $f\colon{\cal P}_{\cal A} \rightarrow \mathbb{R}_{\geq 0}$. Then the \emph{differential entropy} of $P$ is defined as
\begin{equation} \label{eq:de}
\recht{h}(P) = \int_{p \in [0,1]^a} \delta\left(\sum_x p_x - 1\right)f(p)\recht{log}\frac{1}{f(p)} \textrm{d}p.
\end{equation}
If $p \in [0,1]$, we let $\recht{H}_{\recht{b}}(p)$ be the entropy of a Bernoulli random variable with parameter $p$, i.e.  $\recht{H}_{\recht{b}}(p) = - p \log p - (1-p) \log (1-p)$. For more background on information-theoretic definitions we refer to \cite{cover2012elements,csiszar2011information}.

\subsection{The LDP setting}
\label{ssec:ldp}

We consider the setting of \emph{Local Differential Privacy}. 
There are $n$ users, and user $i$ has a private data item $X_i\in\cal A$. 
The aggregator publishes a privacy protocol $Q = (\tilde{Q},\mathcal{B})$. 
User $i$ calculates $Y_i := Q(X_i)$ and sends $Y_i$ to the aggregator. 
We write $\vec Y=(Y_i)_{i=1}^n$; this is the data the aggregator has at their disposal. 
Note that the LDP setting differs from DP in that the users do not give the aggregator their bare private data.

In order to apply information theory, we describe everything in a probabilistic way. 
We consider $X_1,\cdots,X_n$ to be drawn independently from a probability distribution $p \in \mathcal{P}_{\mathcal{A}}$. 
The probability distribution $p$ is unknown to the aggregator as well: 
hence we consider it to be a random variable $P$. 
The distribution of $P$ is given by a probability measure $\mu$ on $\mathcal{P}_{\mathcal{A}}$ (with respect to the Borel $\sigma$-algebra \cite{halmos2013measure} on $\mathcal{P}_{\mathcal{A}}$); 
this distribution reflects the prior knowledge of the aggregator. 
The extended model is depicted in Fig.\,\ref{fig:ext}. 
We consider $\mu$ with the following two properties:
\begin{eqnarray} 
\mu(\{p \in \mathcal{P}_{\mathcal{A}}: \exists x \textrm{ s.t. } p_x = 1\}) &<& 1, \label{eq:muass1}\\
\#\recht{supp}(\mu) &>& 1 \label{eq:muass2}.
\end{eqnarray}
If (\ref{eq:muass1}) does not hold, then we are in the degenerate situation that the aggregator has prior knowledge that all users have the same value. In such a situation, there is no concept of privacy. 
If (\ref{eq:muass2}) does not hold, then a priori there is only one possible value for $P$. In such a situation, there is no concept of utility.

\begin{figure}
\begin{center}
\begin{tikzpicture}[scale = 0.4]
\shadedraw[inner color=black!10,outer color=black!30, draw=black] (-2,-5) --node[below]{$\mathcal{P}_{\mathcal{A}}$} (3,-5) -- (0.5,-0.6999) -- cycle;
\filldraw (0,-3) circle (2pt) node[below](p) {$P$} ;
\draw (8,0) node[draw, rounded corners](a1) {$X_1$};
\draw (16,0) node[draw, rounded corners](b1) {$Y_1$};
\draw[-,decorate,decoration={snake}] (a1) --node[above]{$Q$} (b1);
\draw (8,-2) node[draw, rounded corners](a2) {$X_2$};
\draw (16,-2) node[draw, rounded corners](b2) {$Y_2$};
\draw[-,decorate,decoration={snake}] (a2) --node[above]{$Q$} (b2);
\draw (8,-6) node[draw, rounded corners](an) {$X_n$};
\draw (16,-6) node[draw, rounded corners](bn) {$Y_n$};
\draw[-,decorate,decoration={snake}] (an) --node[above]{$Q$} (bn);
\draw[-, dashed] (8,-3) -- (8,-5);
\draw[-, dashed] (16,-3) -- (16,-5);
\draw[dashed, rounded corners] (-3,1) --node[above]{Hidden} (10,1) -- (10,-7) -- (-3,-7) -- cycle;
\draw[rounded corners] (20,1) --node[above]{Aggregator} (24,1) -- (24,-7) -- (20,-7) -- cycle;
\draw[-] (b1) -- (20,0);
\draw[-] (b2) -- (20,-2);
\draw[-] (bn) -- (20,-6);
\draw (22,-3) node {$\vec{Y}$};
\draw (8,-8) node{$\in \mathcal{A}$};
\draw (16,-8) node{$\in \mathcal{B}$};
\draw (22,-8) node{$\in \mathcal{B}^n$};
\draw[-,decorate,decoration={snake}] (0,-3) -- (a1);
\draw[-,decorate,decoration={snake}] (0,-3) -- (a2);
\draw[-,decorate,decoration={snake}] (0,-3) -- (an);
\end{tikzpicture}
\caption{\it Extended model of the LDP setting, including the stochastic $P$.} 
\label{fig:ext}
\end{center}
\end{figure}
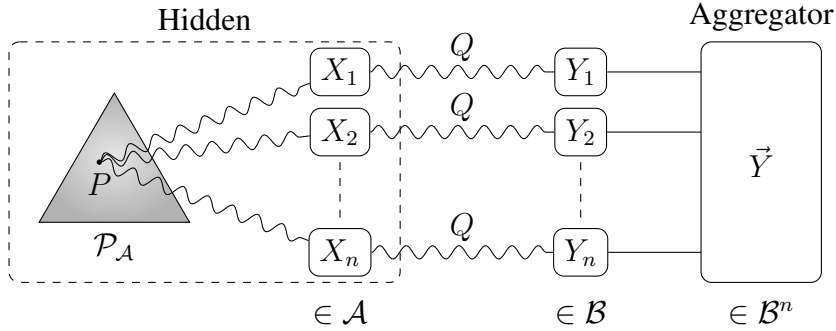

We consider two possible objectives of the aggregator:
\begin{enumerate}[wide=0pt]
\item 
The aggregator wants to know $P$ as accurately as possible. 
For instance, the aggregator is a scientist whose userbase is a sample of a greater population. 
The aggregator is not concerned with this specific userbase per se, but rather with the characteristics of the general population.
\item 
For $\vec{x} \in{\cal A}^n$ and $\gamma \in \mathcal{A}$, let $t_\gamma(\vec{x})=\#\{i: x_i=\gamma\}$
and $t(\vec{x})=(t_\gamma(\vec{x}))_{\gamma\in\cal A}$.
Then the aggregator wants to know the tallies
$T := t(\vec{X})$ as accurately as possible. 
For instance, the users are customers of a service, and the service provider wants to know statistics about this specific group.
\end{enumerate}
We will see later that although these two goals are not directly interchangeable, these goals will overlap more as the number of users grows to infinity. For both of these goals, it suffices for the aggregator to consider $S_y := \#\{i: Y_i = y\}$, for $y \in \mathcal{B}$.

\subsection{Dirichlet distributions} \label{ssec:dirichlet}

An important characteristic of the setting is the prior distribution $\mu$. A typical choice for the prior distribution is the Dirichlet distribution. For a vector $\alpha \in \mathbb{R}^a_{\geq 0}$, the \emph{Dirichlet distribution with parameter $\alpha$} is the continuous probability distribution on $\mathcal{P}_{\mathcal{A}}$ whose probability density function $\Delta_{\alpha}$ is given by $\Delta_{\alpha}(p) = \recht{B}(\alpha)^{-1}\prod_{x=1}^a p_x^{\alpha_x-1}$. Here $\recht{B}$ is the multivariate beta function, i.e.
\begin{equation}
    \recht{B}(\alpha_1,\cdots,\alpha_a) = \frac{\prod_{x=1}^a \Gamma(\alpha_x)}{\Gamma\left(\sum_{x=1}^a \alpha_x\right)}.
\end{equation}
In the situation where the aggregator does not have any prior information, a reasonable choice for the prior distribution $\mu$ is the \emph{Jeffreys prior}. Its key characteristic is that it is an uninformative prior, i.e. its formulation does not depend on the parametrisation of the parameter space. In our situation, where the parameter space is $\mathcal{P}_{\mathcal{A}}$, the Jeffreys prior equals the Dirichlet distribution with parameter vector $(\tfrac{1}{2},\cdots,\tfrac{1}{2})$.

\section{Existing metrics}

\subsection{Current non-information-theoretic metrics}

\subsubsection{Privacy}

The de facto standard for assessing the level of privacy of a protocol is \emph{Local Differential Privacy}. 
It was introduced to be the `local' counterpart to DP \cite{dwork2006calibrating}. 

\begin{definition} \label{def:ldp}
Let $Q$ be a privacy protocol for $\mathcal{A}$. Then the \emph{local differential privacy level} \cite{kasiviswanathan2011can} of $Q$ is defined as
\begin{equation}
\recht{LDP}(Q) := \max_{x,x' \in \mathcal{A}} \max_{y \in \mathcal{B}} \recht{log}\frac{\tilde{Q}_{y|x}}{\tilde{Q}_{y|x'}}.
\end{equation}
Let $\varepsilon \in \mathbb{R}_{\geq 0}$. We say that $Q$ \emph{satisfies $\varepsilon$-LDP} if $\recht{LDP}(Q) \leq \varepsilon$.
\end{definition}

A lower LDP level ensures higher privacy, as the probability distributions $\tilde{Q}_{\bullet|x}$ and $\tilde{Q}_{\bullet|x'}$ do not differ too much; 
this makes it hard for the aggregator to determine $X_i$ given $Q(X_i)$. While the notion of $\varepsilon$-LDP gives strong privacy guarantees, even in worst-case settings, it has several drawbacks. First, it has been noted that DP, the counterpart of LDP in the centralised setting, is too strict for many applications, because it puts a strong requirement on the relative values of output probabilities, even for those $y$ which occur with only a very low probability. As such many relaxations for DP exist \cite{dwork2006our,cuff2016differential,cuff2016differential,dwork2016concentrated,mironov2017renyi}. Second, $\varepsilon$-LDP is defined only for probabilistic protocols; in particular, deterministic protocols do not offer privacy according to the LDP metric (see Example \ref{ex:ldp}). 
    Deterministic protocols are used, however, in the context of, for example $k$-anonimity \cite{samarati1998protecting}.
    Although it is generally felt that such protocols are not very good, it would be nice to have a metric that is able to measure their quality.
\begin{example} \label{ex:ldp}
Suppose $\mathcal{A} = \{1,\cdots,2k\}$ for an integer $k$, and $Q$ is the deterministic privacy protocol given by $Q(x) = x \mod 2$. Then $\recht{LDP}(Q) = \infty$, because $\tilde{Q}_{1|1} = 1$ whereas $\tilde{Q}_{1|0} = 0$; hence in the LDP-metric $Q$ does not offer any privacy. On the other hand, intuitively it is clear that $Q$ protects at least some of the user's private information.
\end{example}

\subsubsection{Utility}

The main way of measuring the utility of a protocol is by looking at frequency estimators, and how well they estimate the original input frequencies.

\begin{definition} 
\label{def:accnorm}
Let $\varphi \in \mathbb{R}^{\mathcal{A}}$ be the real (hidden) vector of the frequencies of the different data items among the user population, and let $\hat\varphi\colon \mathcal{Y}^n \rightarrow \mathbb{R}^{\mathcal{A}}$ be a function that computes 
an estimator of $\varphi$ from~$\vec{Y}$. 
Let $d\colon \mathbb{R}^{\mathcal{A}} \rightarrow \mathbb{R}_{\geq 0}$ be a metric on $\mathbb{R}^{\mathcal{A}}$. Then the \emph{$d$-accuracy} of $(Q,\hat\varphi)$ for $n$ users given $\varphi$ is defined as
\begin{equation}
    \recht{Acc}_{d,n,\varphi}(Q,\hat\varphi) = \mathbb{E}\left[d(\varphi,\hat\varphi(\vec{Y}))\right].
\end{equation}
\end{definition}

Typically one chooses the $\ell_1$- or $\ell_2$-norm for $d$. 
The metric depends on $d$, $n$, and $\varphi$, but this does not matter too much: 
we can get rid of the dependence on $\varphi$ by taking either the expected value over the 
prior distribution or the worst case \cite{blasiok2019towards}, 
and the choice of $d$ can be tailored to the aggregator's demands. 
As to the dependence on $n$, for $d = \ell_1,\ell_2$ 
it can be shown that the accuracy can be expressed as $c \cdot n^{-1/2}$ \cite{duchi2013localb,duchi2013local,wang2017locally}, 
hence maximising the accuracy comes down to maximising the constant $c$. 
Nevertheless, there are several issues with this approach. First,
    the metric measures the utility of the pair $(Q,\hat\varphi)$, 
    rather than just the utility of $Q$. 
    This is an issue because in general there are many different choices for $\hat\varphi$, 
    without a clear a priori reason to prefer any one of them \cite{wang2019consistent}. Second, while for explicit linear $\hat\varphi$ the accuracy can be approximated analytically \cite{wang2017locally}, 
    for many other choices of estimators this is more difficult. 
    As such, this metric is hard to study from an analytical perspective, especially its dependence on $\varphi$ and $n$. This makes it hard to compare privacy protocols.

\subsection{Current information-theoretic metrics} \label{ssec:itmetrics}

Many information-theoretic metrics are discussed in \cite{wagner2018technical}; 
in the overview below we restrict ourselves to those that have been used in the context of (local) DP. 
Note that in general a metric for the centralised DP setting can also be applied to the LDP setting, 
since we can view LDP as the repeated application of a DP protocol on single-user databases.

\cite{mir2012information,du2012privacy,makhdoumi2014information,cuff2016differential,wang2016relation}  propose leakage metrics which focus on the mutual information between input and output (either in the localised or in the centralised case), or equivalently, privacy metrics which focus on the uncertainty of the input given the output. Similar approaches to privacy are used in the setting of biometric security \cite{lai2010privacy} and genomic privacy \cite{humbert2013addressing}. \cite{alvim2011relation,barthe2011information,issa2016operational,rogers2016max} define privacy metrics based on min-entropy, representing a worst-case privacy loss, and bound this from above by $\varepsilon$-DP.\footnote{See \cite{fehr2014conditional,verdu2015alpha,ilic2017general} for an overview of the many ways to define conditional entropy and mutual information analogous to R\'{e}nyi entropy, 
of which min-entropy is a specific instance.}  
Continuing in this line, \cite{alvim2011differential} uses a min-entropy utility measure 
and gives an upper bound on utility in terms of $\varepsilon$-DP.  
R\'{e}nyi-DP unifies entropic and min-entropic privacy \cite{mironov2017renyi}.

Other authors \cite{kairouz2014extremal,wang2016mutual} use the mutual information between input and output as a utility metric rather than a privacy metric, and give optimal privacy protocols for a given level of DP. Other information-theoretic approaches to privacy involve the Kullback-Leibler distance between output distributions \cite{duchi2013local,kairouz2014extremal} and distortion theory \cite{du2012privacy,sarwate2014rate,wang2016relation}.

A shared characteristic of the (min-)mutual information privacy and utility metrics is that \emph{any} information obtained by the aggregator is considered to contribute to utility, while \emph{any} information hidden from the aggregator is considered to contribute to privacy. 
By definition, this leads to a situation where any gain in privacy automatically leads to a loss in utility and vice versa. 
However, privacy and utility deal with fundamentally different domains of information \cite{li2009tradeoff}. 
Utility is an \emph{aggregate} construct, and a utility metric should reflect to what extent the aggregator is able to recover information about the userbase as a group. 
By contrast, privacy is an \emph{individual} construct, and a privacy metric should reflect to what extent an individual user is able to hide their private information. 
Since most information-theoretic metrics from the literature do not make the distinction between aggregate and individual information, they do not accurately capture the notions of privacy or utility. An exception is Mutual Information-Differential Privacy \cite{cuff2016differential} in the central setting, which uses conditioned mutual information because, as the authors argue, a strong adversary which may have access to all users' data except for one is implicit in $\varepsilon$-DP. Our privacy metric in Def.~\ref{def:priv} adapts this idea to the local setting.

\section{New metrics}

\subsection{A case for average case} \label{ssec:average}

Many privacy metrics in (L)DP focus on a worst-case analysis. 
This can mean either looking at the worst possible output \cite{dwork2006calibrating}, worst possible difference between two inputs \cite{mironov2017renyi}, or worst possible distribution on the input set \cite{cuff2016differential}. 
The appeal is clear, as a worst case approach provides privacy guarantees that always hold. 
Nevertheless, we believe that it is useful to also
study privacy metrics in which the constraints are somewhat relaxed, i.e. based on leakage analysis 
that is more average-case than worst-case \cite{gotz2011publishing,kasiviswanathan2011can,mironov2017renyi}.
Such a metric can give a high privacy score even if there is the occasional occurrence of large leakage. In this way, we can get an accurate picture of a protocol's overall privacy behaviour, rather than just the worst-case guarantees.

We believe that an average-case approach can be a useful tool alongside $\varepsilon$-LDP for the following reasons. If worst-case privacy is one's focus, an alternative to LDP would be to make use of cryptographic solutions, e.g. Secure Multiparty Computation (SMC).
Protocols for tallying secret values are well studied in the field of SMC \cite{schoenmakers1999simple,long2018distributed,alter2018computing}; they are not computationally expensive but require more communication than LDP.
While LDP protocols always have nonzero leakage, SMC can offer strong privacy guarantees in many situations.\footnote{Note that SMC outputs the true tallies, and as such does not protect a user's privacy in the unlikely case that all other users collude with the aggregator. Note also that SMC and (L)DP are not mutually exclusive, and one may combine protocols to offer different kinds of data protection \cite{acar2017achieving,he2017composing}.} Since LDP invariably involves some loss of privacy, it is beneficial to determine this loss for the typical user.

Furthermore, there may already be leakage about $X_i$ through other channels than the reporting of~$Y_i$. In practice, the aggregator can have access to information about a user's device or geographical settings \cite{tang2017privacy}. Also, the existence of correlations between users' data allows the aggregator to use multiple users' reports to obtain more accurate information about a single user's data \cite{kifer2011no,cuff2016differential}; data from the same user over time can be used in the same way \cite{tang2017privacy}. A malicious aggregator could also use side information to perform timing, state, or privacy budget attacks \cite{haeberlen2011differential}. Therefore, it is interesting to know the effect of a protocol on a typical user's privacy alongside the worst-case guarantees, since in practice an attacker might have the capability of bypassing those guarantees to some extent.

\subsection{New metrics}

We can formalise the distinction between aggregate and individual information discussed in section \ref{ssec:itmetrics} using the language of information theory. The total information in the hands of the users is $\recht{H}(\vec{X})$, while the part of this that the aggregator has at their disposal is $\recht{I}(\vec{Y};\vec{X})$. We can split both of them up into three parts: the \emph{nonprivate part} (the knowledge about $P$), the \emph{group-private part} (the knowledge about $T$ given $P$), and the \emph{individual private part} (the knowledge about $\vec{X}$ given $T$); this is shown in Table \ref{tab:split}. As in the Table, these parts can be combined in three ways:

\begin{table} 
\centering
\begin{tabular}{lccc} 
& Nonprivate & Group-private & Individual private \\ 
\cline{2-4}
Users & \multicolumn{1}{|c|}{$\recht{I}(\vec{X};P)$} & \multicolumn{1}{c|}{$\recht{H}(T|P)$} & \multicolumn{1}{c|}{$\recht{H}(\vec{X}|T)$} \\ \hline \hline
Total information & \multicolumn{3}{|c|}{$\recht{H}(\vec{X})$}  \\ \cline{2-4}
Private information & & \multicolumn{2}{|c|}{$\recht{H}(\vec{X}|P)$}  \\ \cline{2-4}
Tally-relevant & \multicolumn{2}{|c|}{$\recht{H}(T)$} & \\ \cline{2-3}
Distribution-relevant & \multicolumn{1}{|c|}{$\recht{I}(\vec{X};P)$} & &  \\ \cline{2-2}
& \\ \cline{2-4}
Aggregator & \multicolumn{1}{|c|}{$\recht{I}(\vec{Y};P)$} & \multicolumn{1}{c|}{$\recht{I}(\vec{Y};T|P)$} & \multicolumn{1}{c|}{$\recht{I}(\vec{Y};\vec{X}|T)$} \\ \hline \hline
Total obtained information & \multicolumn{3}{|c|}{$\recht{I}(\vec{Y};\vec{X})$} \\ \cline{2-4}
Privacy leakage  & & \multicolumn{2}{|c|}{$\recht{I}(\vec{Y};\vec{X}|P)$} \\ \cline{2-4}
Tally gains & \multicolumn{2}{|c|}{$\recht{I}(\vec{Y};T)$} & \\ \cline{2-3}
Distribution gains & \multicolumn{1}{|c|}{$\recht{I}(\vec{Y};P)$} & & \\ \cline{2-2}
\end{tabular}
\caption{The information split in the LDP setting.} \label{tab:split}
\end{table}

\begin{enumerate}[wide=0pt]
    \item The group-private part and the individual private part together form the private information $\recht{H}(\vec{X}|P)$ on the side of the users. 
    We do not consider all information held by the users to be private, as revealing the distribution $P$ does not violate any user's privacy. 
To see this, suppose the tallies $T$ for a group of $n$ users are revealed. 
If $n = 1$, the single user's privacy is completely violated. 
However, as $n$ grows larger, the amount of private information about a single user that can be obtained from $T$ decreases. 
In the limit $n\to\infty$, the tallies contain no private information about a single user at all; 
but $P$ represents the `tallies' of an infinitely large population from which our users are a sample. Hence the total private information is $\recht{H}(\vec{X}|P)$, while the private information available to the aggregator (i.e. the privacy leakage) equals $\recht{I}(\vec{Y};\vec{X}|P)$. Note that this mirrors the definition of leakage in the central setting in \cite{cuff2016differential}.
\item The nonprivate part and the group-private part together form the information about tallies $\recht{H}(T)$. Its counterpart on the aggregator's side, $\recht{I}(\vec{Y};T)$, reflects how much the aggregator learns about the tallies. Note that $\recht{I}(\vec{Y};T) = \recht{I}(S;T)$, where $S$ is as in section \ref{ssec:ldp}.
\item In the same way, the nonprivate part is the information about $P$ available to the users or the aggregator. Similar to the above, we have $\recht{I}(\vec{Y};P) = \recht{I}(S;P)$.
\end{enumerate}
We define a normalised notion of privacy by defining privacy as the fraction of private information that is not leaked to the aggregator:

\begin{definition} 
\label{def:priv}
Let $Q$ be a privacy protocol.
Let $\mu$ be the aggregator's prior.
We define the \emph{average privacy} of $Q$ as follows:
\begin{equation}
\recht{S}_{\mu}(Q) :=
1-\frac{\recht{I}(\vec{Y};\vec{X}|P)}{\recht{H}(\vec{X}|P)} = \frac{\recht{H}(\vec{X}|\vec{Y},P)}{\recht{H}(\vec{X}|P)}.
\end{equation}
\end{definition}
The dependence on $\mu$ enters via the conditioning on $P$.
Note that the denominator in the fraction above is nonzero by  (\ref{eq:muass1}). 
An important attribute of this privacy metric is that although its definition depends on the number of users, its value does not:

\begin{lemma} 
\label{lem:priv}
The average privacy $\recht{S}_{\mu}(Q) = \frac{\recht{H}(\vec{X}|\vec{Y},P)}{\recht{H}(\vec{X}|P)}$ does not depend on $n$.
\end{lemma}

In the same way, we define utility as the fraction of relevant information that is available to the aggregator:

\begin{definition} 
\label{def:uti}
Let $n$ be the number of users, and let $Q$ be a privacy protocol. We define the \emph{distribution utility} and \emph{tally utility} of $Q$ as follows:
\begin{equation}
\recht{U}^{\recht{distr}}_{n,\mu}(Q) := \frac{\recht{I}(\vec{Y};P)}{\recht{I}(\vec{X};P)}; \ \ \
\recht{U}^{\recht{tally}}_{n,\mu}(Q) := \frac{\recht{I}(\vec{Y};T)}{\recht{H}(T)}.
\end{equation}
\end{definition}
The denominators in the fractions above are nonzero by property (\ref{eq:muass2}).
Defs. \ref{def:priv} and \ref{def:uti}
have a number of desirable properties:
\begin{enumerate}[wide=0pt]
    \item 
    Both utility and privacy have a clear intuitive meaning as the amount of relevant information obtained by, or hidden from, the aggregator. 
    \item 
    The utility metrics do not depend on the choice of estimator.
    \item 
    In contrast to other information-theoric metrics, our metrics take into account the fact that not all obtained information is relevant to the aggregator, and not all revealed information violates a user's privacy.
    \item 
    Our privacy metric can also be applied in deterministic contexts. 
    For instance, in Example \ref{ex:ldp}, assuming a symmetric Dirichlet prior on $\mathcal{P}_\mathcal{A}$ with parameter $\tfrac{1}{2}$ (the Jeffreys prior, see section \ref{sec:protocols}), we find 
    \begin{equation}
        \recht{S}_{\mu}(Q) = \frac{\dg\left(\frac{k+2}{2}\right)-\dg\left(\frac{3}{2}\right)}{\dg\left(k+1\right)-\dg\left(\frac{3}{2}\right)} > 0,
    \end{equation}
    where $\dg$ is the digamma function. 
    In contrast,  according to the LDP metric this protocol does not offer any privacy at all. 
    See also Example \ref{ex:lh}.
    \item 
    Our metrics look at the average case. 
    This means that our privacy metric will be less `strict' than LDP, which looks at worst-case privacy; this will be formalised in Theorem \ref{thm:wc}.
    \item 
    Like the estimator-based metrics, our utility metrics depend on the number of users~$n$.
    This makes it difficult to compare protocols.
    However, in Theorem \ref{thm:iypiytp} we will show that a protocol can be characterised by its asymptotic $n\to\infty$ behaviour 
    $\recht S_\mu^{\rm stable}$, which we can explicitly compute. 
    
\end{enumerate}
The metrics can be computationally involved, especially for large $n$, $a$ and $b$. 
In Section~\ref{sec:protocols} we discuss the computational complexity of the various metric, both in general cases and for specific protocols.

\begin{remark}
In \cite{kairouz2014extremal} the authors study the case that $P$ can take only two possible values, and they describe a method that is optimal in distinguishing these two cases, given the obfuscated data of one user. In our language, this means that the support of $\mu$ contains two points. The metric $\recht{S}_{n,\mu}$ generalises this scenario\footnote{The actual metric is different, as \cite{kairouz2014extremal} focuses on Kullback-Leibler divergence rather than mutual information.} in two ways: we allow any prior distribution $\mu$, and any number of users.
\end{remark}

\begin{remark}
In Table \ref{tab:split} one finds the identity $\recht{I}(\vec{Y};\vec{X}) = \recht{I}(\vec{Y};P) + \recht{I}(\vec{Y};\vec{X}|P)$. Since the first term is related to (distribution) utility, and the second term is related to privacy, one might be led to believe that this implies a sum law between privacy and utility. However, this is not the case, as $\recht{I}(\vec{Y};\vec{X}|P)$ expresses the leakage, rather than the privacy, of the protocol $Q$. As such, this formula expresses the idea ``utility $+$ leakage $=$ total information", which does not imply a relation between utility and privacy, as the term $\recht{I}(\vec{Y};\vec{X})$ also depends on $Q$.
\end{remark}

\subsection{Digit utility}

The utility metrics introduced in the previous section are relative to the total amount of relevant information in the unobfuscated data. 
This makes them useful for comparing protocols to the ideal functionality and to each other. 
Nevertheless, in some cases it is useful to have a utility metric that measures the aggregator's information gain in an absolute way; 
we will now define such a metric for the case that the aggregator is interested in $P$, and $\mu$ is a continuous probability distribution. 
The distribution $P$ itself then carries infinite information, but the aggregator can obtain only a finite amount of this. 
One way to quantify the amount of information  is to determine how many digits of $P$ the aggregator has learned. 
This leads us to the following definition:

\begin{definition}
Let $Q$ be a faithful privacy protocol, and let the prior $\mu$ be a continuous measure on $\mathcal{P}_{\mathcal{A}}$. 
We define the \emph{digit utility} of $Q$ as
\begin{equation} \label{eq:udigit}
\recht{U}^{\recht{digit}}_{n,\mu}(Q) := \frac{\recht{I}(\vec{Y};P)-\recht{h}(P)}{a-1}.
\end{equation}
\end{definition}
The intuition behind the digit utility is that it approximates how many digits (in base $\textrm{e}$) of $P$
the aggregator has learned. Note that this is essentially a rescaling of $\recht{U}^{\recht{distr}}_{n,\mu}(Q)$, and that it does not carry extra information about $Q$.

\begin{lemma} \label{lem:digit}
For $d \in \mathbb{R}_{\geq 0}$, let $P_{\langle d \rangle}$ be the discretisation of $P$ with coordinate-wise step size $\textrm{\emph{e}}^{-d}$, i.e. centered on the lattice $(\textrm{\emph{e}}^{-d}\mathbb{Z})^a$. Then
\begin{equation}
\recht{lim}_{n \rightarrow \infty} \left(\recht{I}(\vec{Y};P) - \recht{H}(P_{\langle \recht{U}^{\recht{digit}}_{n,\mu}(Q) \rangle})\right) = 0.
\end{equation}
\end{lemma}

In other words, for large $n$ the amount of information the aggregator has about $P$ is approximately equivalent to the amount of information gained from learning the first $\recht{U}^{\recht{digit}}_{n,\mu}(Q)$ digits of $P$. 
The digit utility thus provides an information-theoretic analogon to the accuracy utility metric in Def. \ref{def:accnorm}. It has the additional advantage that we do not need to specify a distance measure or an estimator. 
It is a useful metric if the aggregator is interested in knowing $P$ up to some established accuracy. 
After the aggregator has chosen a protocol, the digit utility can be used to tell him how many users are needed to obtain the desired accuracy.

\section{Worst case privacy} \label{sec:wc}

In our formalism, a measure for
worst-case privacy is obtained by taking the infimum of the hidden information over 
all outcomes $y \in \mathcal{B}$ and all prior distributions~$\mu$:

\begin{definition} \label{def:wc}
Let $i\in\{1,\ldots,n\}$ be arbitrary, and let $\mathcal{M}$ be the set of all probability measures on $\mathcal{P}_{\mathcal{A}}$ satisfying (\ref{eq:muass1}). Let $Q$ be a privacy protocol. Then we define the \emph{worst-case privacy} of $Q$ to be
\begin{equation}
    \recht{S}^{\recht{wc}}(Q) := \recht{inf}_{\mu \in \mathcal{M}} \recht{inf}_{y \in \mathcal{B}} \frac{\recht{H}(X_i|Y_i = y,P)}{\recht{H}(X_i|P)}.
\end{equation}
\end{definition}
In other words, $\recht{S}^{\recht{wc}}(Q)$ is the fraction of the user's private information that remains hidden from the aggregator, when assuming the worst possible prior distribution $\mu$ and user output $y$. 
Trivially the following result holds:

\begin{lemma} \label{lem:wc}
Let $Q$ be a privacy protocol. Then $\recht{S}_{\mu}(Q) \geq \recht{S}^{\recht{wc}}(Q)$ for any prior measure $\mu$.
\end{lemma}

It turns out that this notion of worst-case privacy is directly related to LDP:

\begin{theorem} 
\label{thm:wc}
Let $Q$ be a privacy protocol. Then $\recht{S}^{\recht{wc}}(Q) = \recht{e}^{-\recht{LDP}(Q)}$.
\end{theorem}

This theorem shows that LDP has a direct information-theoretic interpretation as a measure of worst case privacy.

\begin{remark}
Alternatively, one could define worst-case privacy as $\recht{inf}_{p \in \mathcal{P}_{\mathcal{A}}} \recht{inf}_{y \in \mathcal{B}} \frac{\recht{H}(X_i|Y_i = y, P =  p)}{\recht{H}(X_i|P = p)}$. It turns out that this is equivalent to Def. \ref{def:wc}, see appendix \ref{app:wc}. The worst case is as follows: let $y,x,x'$ be such that $\frac{\tilde{Q}_{y|x}}{\tilde{Q}_{y|x'}}$ is maximal, i.e. $\frac{\tilde{Q}_{y|x}}{\tilde{Q}_{y|x'}} = \textrm{e}^{\recht{LDP}(Q)}$. For $\varepsilon \in (0,1)$, let $p^{\varepsilon} \in \mathcal{P}_{\mathcal{A}}$ be of the form $p^{\varepsilon}_x = 1-\varepsilon$, $p^{\varepsilon}_{x'} = \varepsilon$, and all other probabilities equal to $0$. Then we obtain the infimum when $Y_i = y$, and we let $P = \recht{lim}_{\varepsilon \rightarrow 0} p^{\varepsilon}$. In other words, the most information will be leaked when a user outputs an $y$ for which the $x$ that maximises $\tilde{Q}_{y|x}$ occurs with probability almost $1$. In this case, the combination of the knowledge about the prior, plus the fact that the user outputs $y$, allows the aggregator to conclude $X_i = x$ with optimal certainty.
\end{remark}

\begin{remark} 
\label{rem:wc}
Theorem~\ref{thm:wc} gives a better understanding of Example \ref{ex:ldp}. The `reason' that $\recht{LDP}(Q) = \infty$ here (or equivalently $\recht{S}^{\recht{wc}}(Q) = 0$) is that there is a prior under which $Q$ reveals all information; for example, if the support of $\mu$ is contained in $\mathcal{P}_{\{1,2\}} \subset \mathcal{P}_{\{1,\cdots,2k\}}$ (so that it is known a priori that users can only have values $1$ and $2$ as their private data), then $X_i$ becomes known to the aggregator when he learns $Y_i$.
\end{remark}

\section{Limit behaviour of the utility} 
\label{sec:limit}

Throughout this section we assume that the prior $\mu$ is \emph{continuous}. Although our utility metrics depend on the number of users $n$, it turns out that their asymptotic behaviour can be explicitly calculated. The main result is formulated in Theorem \ref{thm:iypiytp}.

\subsection{Faithful protocols}

Before we dive into metrics we want to define a certain class of privacy protocols that will turn out to be `optimal' with respect to utility for asymptotically many users.

\begin{definition}
Let $Q$ be a privacy protocol for $\mathcal{A}$. We say $Q$ is \emph{faithful} if $\recht{rk}(\tilde{Q}) = a$.
\end{definition}

The importance of faithfulness is highlighted by the following Proposition.

\begin{proposition} \label{prop:faithful}
Let $Q$ be a privacy protocol for $\mathcal{A}$. Then $Q$ is faithful if and only if the map $Q_*\colon \mathcal{P}_{\mathcal{A}} \rightarrow \mathcal{P}_{\mathcal{B}}$ is injective.
\end{proposition}

Informally, this Proposition tells us that if $Q$ is not faithful, then there will be $p, p' \in \mathcal{P}_A$ for which the aggregator will not be able to distinguish $P = p$ from $P = p'$. 
For this reason non-faithful privacy protocols only have a limited role in our discussions; 
indeed, the main privacy protocols from the literature are all faithful (they come with unbiased estimators). 
Nevertheless, non-faithful privacy protocols are still useful in the situation where the support of the prior  $\mu$ is limited. 
For example, in \cite{kairouz2014extremal} the situation is studied where there are only two possibilities for $P$, and indeed one of the main results is that in the high privacy regime the optimal privacy protocol is not faithful.

\subsection{Limit behaviour} \label{ssec:lim}

We need three more definitions before we can state the main result.

\begin{definition} \label{def:sim}
Let $Q = (\tilde{Q},\mathcal{B})$ be a privacy protocol on $\mathcal{A}$. Let $\mathcal{B}_{>0} = \{y \in \mathcal{B}: \exists x \textrm{ s.t. } \tilde{Q}_{y|x} > 0\}$. We let $\sim$ be the equivalence relation on $\mathcal{B}_{>0}$ generated by the relation $\{(y,y') : \exists x \textrm{ s.t. } \tilde{Q}_{y|x} > 0, \tilde{Q}_{y'|x} > 0\}$. We write $\mathcal{B}_{>0}/{\sim}$ for the set of equivalence classes under this equivalence relation.
\end{definition}

\begin{definition}
Let $A_1,A_2,\cdots$ and $B_1,B_2,\cdots$ be infinite sequences of real numbers. We write $A_n \stackrel{\infty}{\longrightarrow} B_n$ if $\lim_{n \rightarrow \infty} (A_n-B_n) = 0$.
\end{definition}

\begin{definition} \label{def:uas}
Let $Q = (\tilde{Q},\mathcal{B})$ be a faithful privacy protocol for $\mathcal{A}$. For $p \in \mathcal{P}_{\mathcal{A}}$,  let $D_p := \recht{diag}\left(\tfrac{1}{(\tilde{Q}p)_1},\cdots,\tfrac{1}{(\tilde{Q}p)_b}\right)$.
We define the \emph{asymptotic utility} of $Q$ to be
\begin{equation}
\recht{U}^{\recht{as}}_{\mu}(Q) = -\frac{1}{2}\log(2\pi\textrm{e}) + \frac{1}{2a-2}\mathbb{E}_{P}\log \recht{det}(\tilde{Q}^{\recht{T}}D_P\tilde{Q}).
\end{equation}
\end{definition}

Furthermore, we define $\recht{C}_{\mu} := \recht{U}^{\recht{as}}_{\mu}(\recht{id}_{\mathcal{A}}) = -\tfrac{1}{2}\log(2\pi\textrm{e})-\tfrac{1}{2a-2}\mathbb{E}_P \sum_x \log P_x$. It turns out that this is an upper bound on asymptotic utility.

\begin{lemma} 
\label{lem:uas}
For any faithful privacy protocol $Q$ on $\mathcal{A}$ one has $\recht{U}^{\recht{as}}_{\mu}(Q) \leq \recht{C}_{\mu}$.
\end{lemma}

This allows us to formulate the following result, which describes the asymptotic behaviour of our utility measures, and motivates the terminology `asymptotic utility':

\begin{theorem} 
\label{thm:iypiytp}
Let $Q$ be a privacy protocol, and let $d := \recht{rk}(\tilde{Q})$.
Let $\mathcal{B}_{>0}/{\sim}$ be as in Def. \ref{def:sim}, and let $b' := \#(\mathcal{B}_{>0}/{\sim})$.
Then
\begin{enumerate}
    \item 
    There exists a constant $r_\mu(Q)$ such that $\recht{I}(\vec{Y};P) \stackrel{\infty}{\longrightarrow} \frac{d-1}{2}\recht{log} n + r_{\mu}(Q)$. If $Q$ is faithful, then $r_{\mu}(Q) = (a-1)\recht{U}^{\recht{as}}_{\mu}(Q)+\recht{h}(P)$.
    \item   There exists a constant $s_{\mu}(Q)$ such that $\recht{I}(\vec{Y};T) \stackrel{\infty}{\longrightarrow} \frac{d+b'-2}{2}\recht{log} n + s_{\mu}(Q)$.
\end{enumerate}
\end{theorem}

From this we can derive the following result about the limit behaviour of utility metrics:

\begin{corollary}
\label{cor:limuti}
Let $Q$ be a privacy protocol. 
Let $b'=\#(\mathcal{B}_{>0}/{\sim})$
and $d=\recht{rk}(\tilde{Q})$.
Then for every privacy protocol $Q$ the following holds as $n \rightarrow \infty$:
\begin{eqnarray}
\recht{U}^{\recht{distr}}_{n,\mu}(Q) &\stackrel{\infty}{\longrightarrow}& \frac{d-1}{a-1}; \label{eq:limuti_1} \\
\recht{U}^{\recht{tally}}_{n,\mu}(Q) &\stackrel{\infty}{\longrightarrow}& \frac{d+b'-2}{2a-2}. \label{eq:limuti_2}
\end{eqnarray}
Furthermore, if $Q$ is faithful then
\begin{eqnarray}
\recht{U}^{\recht{digit}}_{n,\mu}(Q) &\stackrel{\infty}{\longrightarrow}& \frac{1}{2}\recht{log}(n)+\recht{U}^{\recht{as}}_{\mu}(Q); \label{eq:limuti_3}
\\
    1-\recht{U}^{\recht{distr}}_{n,\mu}(Q) & \stackrel{\infty}{\longrightarrow} & 2\frac{\recht{C}_{\mu}-\recht{U}^{\recht{as}}_{\mu}(Q)}{\log n}. \label{eq:limuti_4}
\end{eqnarray}
\end{corollary}

This Corollary shows that 
for asymptotically large $n$
faithful protocols (which have $d = a$) have higher utility than other protocols. 
Also, note that privacy protocols with a high value of $\recht{U}^{\recht{as}}_{\mu}(Q)$ have both good digit and distribution utilities for asymptotically large $n$; this motivates the term `asymptotic utility'. 

\begin{remark}
$\phantom{a}$
\begin{enumerate}[wide=0pt]
\item 
Typical privacy protocols are faithful, and they satisfy $b' = 1$; one can have $b' > 1$ only if $\recht{LDP}(Q) = 0$. As such, most generally used privacy protocols satisfy $\recht{U}^{\recht{distr}}_{n,\mu} \stackrel{\infty}{\longrightarrow} 1$, $\recht{U}^{\recht{tally}}_{n,\mu} \stackrel{\infty}{\longrightarrow} \frac{1}{2}$.

\item 
The fact that $\recht{U}^{\recht{digit}}_{n,\mu}(Q) \stackrel{\infty}{\longrightarrow} \frac{1}{2}\recht{log}(n)+\recht{U}^{\recht{as}}_{\mu}(Q)$ tells us that for large $n$, we know approximately the first $\frac{1}{2}\recht{log}(n)+\recht{U}^{\recht{as}}_{\mu}(Q)$ digits of $P$, i.e. the accuracy to which we know $P$ is approximately $\textrm{e}^{\recht{U}^{\recht{digit}}_{n,\mu}(Q)} \approx c \cdot n^{-1/2}$. This mirrors the fact that the standard deviation of linear frequency estimators also scales as $n^{-1/2}$ for large $n$ \cite{wang2017locally}.
\end{enumerate}
\end{remark}
The relation between the constant $\recht{U}^{\recht{as}}_{\mu}(Q)$ and the asymptotic behaviour of the protocol $Q$ can also be formulated in a different way:

\begin{definition}
\label{def:participation}
Let $\mu$ be continuous, and let $Q$ be a faithful privacy protocol. 
We define the \emph{effective participation factor} of $Q$ to be $\recht{F}_{\mu}(Q) := \textrm{e}^{2\recht{U}^{\recht{as}}_{\mu}(Q)-2\recht{C}_{\mu}}$.
\end{definition}

Note that by Lemma \ref{lem:uas} one has $\recht{F}_{\mu}(Q) \in (0,1]$. The terminology `effective participation factor' is motivated by the following Proposition:

\begin{proposition} 
\label{prop:fmu}
Let $\mu$ be continuous, and let $Q$ be a privacy protocol.
\begin{enumerate}
\item If $Q$ is faithful, then for $n \rightarrow \infty$ we get $\recht{U}^{\recht{digit}}_{n,\mu}(Q) \stackrel{\infty}{\longrightarrow} \recht{U}^{\recht{digit}}_{\lfloor \recht{F}_{\mu}(Q) \cdot n \rfloor,\mu}(\recht{id}_{\mathcal{A}})$.
\item If $Q$ is not faithful, then for every $\varepsilon > 0$ we have $\recht{U}^{\recht{digit}}_{n,\mu}(Q) < \recht{U}^{\recht{digit}}_{\lfloor \varepsilon \cdot n \rfloor,\mu}(\recht{id}_{\mathcal{A}})$ for all sufficiently large $n$.
\end{enumerate}
\end{proposition}

In words: for large $n$, the aggregator learns as much about $P$ from the {\em obfuscated} data of $n$ users as they would from the {\em un-obfuscated} data of $\recht{F}_\mu(Q) \cdot n$ users. 
The factor $\recht{F}_\mu(Q)$ represents the factor in efficiency the aggregator pays for using the privacy protocol. Because of the second point we define $\recht{F}_\mu(Q) = 0$ for non-faithful $Q$.

To summarise this section, we have shown that $\recht{I}(\vec{Y};P)$ grows as $c_1 \log n + c_2$ for large $n$ and for some constants $c_1$, $c_2$. The constant $c_1$ is maximal if $Q$ is faithful, while the constant $c_2$ can be explicitely calculated for faithful protocols. For faithful $Q$, another way to express their asymptotic utility is by their effective participation factor $\recht{F}_\mu$, which expresses the information in the obfuscated data as equivalent to the information in the true data of a smaller number of users.

\section{Combining protocols} 
\label{sec:combination}

We look at the behaviour of our metrics when 
protocols are combined in different ways.
It turns out that we can recover many properties of more established privacy metrics.

\subsection{Postprocessing}

A key property of privacy is that further processing of the obfuscated data should not decrease the privacy, and it should not increase the utility. Formally we can define the sequential application of multiple protocols as follows.

\begin{definition}
Let $\mathcal{A}$ be a finite set, let $Q = (\tilde{Q},\mathcal{B})$ be a privacy protocol on $\mathcal{A}$, and let $R = (\tilde{R},\mathcal{C})$ be a privacy protocol on $\mathcal{B}$. 
Then the \emph{composition} $R \circ Q$ is the privacy protocol on $\cal A$
defined by $R \circ Q := (\tilde{R} \cdot \tilde{Q}, \mathcal{C})$, where $\cdot$ denotes matrix multiplication.
\end{definition}

It is easily verified that this is the `correct' definition of composition, i.e. for every $x \in \mathcal{A}$ and $z \in \mathcal{C}$ we have $\mathbb{P}((R \circ Q)(x) = z) = \mathbb{P}(R(Q(x)) = z)$. 
Proposition \ref{prop:postprocessing} describes the behaviour of our utility and privacy metrics under composition. 

\begin{definition} 
\label{def:pushforward}
Let $Q_*\colon \mathcal{P}_{\mathcal{A}} \rightarrow \mathcal{P}_{\mathcal{B}}$ be as in section \ref{ssec:priv}. We let $Q_{**}\mu$ be the pushforward probability measure on $\mathcal{P}_{\mathcal{B}}$, i.e. for each measurable $U \subset \mathcal{P}_{\mathcal{B}}$ we have $(Q_{**}\mu)(U) = \mu(Q_*^{-1}U)$.
\end{definition}

\begin{proposition} 
\label{prop:postprocessing}
Let $Q = (\tilde{Q},\mathcal{B})$ by a privacy protocol for $\mathcal{A}$ and let $R = (\tilde{R},\mathcal{C})$ be a privacy protocol for~$\mathcal{B}$. Then
\begin{eqnarray}
\recht{U}^{\recht{distr}}_{n,\mu}(R \circ Q) &\leq & \recht{min}\{\recht{U}^{\recht{distr}}_{n,\mu}(Q),\recht{U}^{\recht{distr}}_{n,Q_{**}\mu}(R)\}, \label{eq:postprocessing_1}\\
\recht{U}^{\recht{tally}}_{n,\mu}(R \circ Q) &\leq & \recht{U}^{\recht{tally}}_{n,\mu}(Q), \label{eq:postprocessing_2}\\
\recht{S}_{\mu}(R \circ Q) &\geq & \recht{S}_{\mu}(Q), \label{eq:postprocessing_3}\\
\recht{S}^{\recht{wc}}(R \circ Q) &\geq & \recht{S}^{\recht{wc}}(Q). \label{eq:postprocessing_4}
\end{eqnarray}
If $\mu$ is continuous we furthermore have 
\begin{equation}
    \recht{U}^{\recht{digit}}_{n,\mu}(R \circ Q) \leq \recht{U}^{\recht{digit}}_{n,\mu}(Q). \label{eq:postprocessing_5}
\end{equation}
\end{proposition}

In this Proposition, (\ref{eq:postprocessing_5}) is equivalent to an analogous result about LDP \cite{dwork2006calibrating}. (\ref{eq:postprocessing_4}) is the analogous statement for the privacy metric $\recht{S}_{\mu}$.

\begin{example} 
\label{ex:lh}
Let $g \in \{1,\cdots,a\}$
and let ${\cal G}=\{1,\ldots,g\}$.
Let $\mathcal{H}$ be the set of hash functions $h\colon \mathcal{A} \rightarrow \mathcal{G}$. 
Let $\varepsilon \in (0,\infty)$. 
\emph{Local Hash} \cite{wang2016relation} is the privacy protocol $\recht{LH}^{a,g,\varepsilon} = (\tilde{Q},\mathcal{H} \times \mathcal{G})$ that can be described as follows: given $x \in \cal A$, we first uniformly pick a random hash function $h \in \mathcal{H}$, and then we perform $\textrm{GRR}^{g,\varepsilon}$ on $h(x)$, (see section \ref{ssec:grr}). 
Intuitively, the data is obscured twice, first by hashing into a smaller set of categories, and then applying the obfuscation protocol GRR. 
One can calculate
\begin{equation}
\recht{S}_{\mu}(\recht{LH}^{a,g,\varepsilon}) = \mathbb{E}_{h}\left[\recht{S}_{\mu}(h)+\frac{\recht{H}(h(X_i)|P)}{\recht{H}(X_i|P)}\recht{S}_{h_{**}\mu}(\recht{GRR}^{g,\varepsilon})\right].
\end{equation}
where $h \in \mathcal{H}$ is drawn from a uniform distribution on $\mathcal{H}$, and $h_{**}\mu$ is the pushforward measure as defined in Def. \ref{def:pushforward}. This shows that the new privacy metric indeed `detects' the effects of both steps. On the other hand, the LDP metric only sees the privacy from the second step, as we have $\recht{LDP}(\recht{LH}^{a,g,\varepsilon}) = \recht{LDP}(\recht{GRR}^{g,\varepsilon}) = \varepsilon.$
\end{example}

\subsection{Privacy budget}

A central concept in (L)DP is that of a privacy budget. It refers to the fact that when performing multiple privacy protocols $Q^1,\cdots,Q^k$ on the same data simultaneously, each revealed outcome $Q^j(X_i)$ will leak some private information about $X_i$, and giving the aggregator all the $Q^j(X_i)$ will leak more information than any individual protocol. Formally, we can define the parallel application of privacy protocols as follows:

\begin{definition}
Let $\mathcal{A}$ be a finite set, and let $Q^1 = (\tilde{Q}^1,\mathcal{B}^1),\cdots,Q^k = (\tilde{Q}^k,\mathcal{B}^k)$ be $k$ privacy protocols on $\mathcal{A}$. 
Then we define the \emph{product} $Q^1 \times \cdots \times Q^k$ to be the privacy protocol $Q^1 \times \cdots \times Q^k := (\tilde{R},\mathcal{C})$, where $\mathcal{C}$ is the Cartesian product of sets $\prod_j \mathcal{B}^j$, and $\tilde{R} \in [0,1]^{\prod_j \mathcal{B}^j \times \mathcal{A}}$ is the matrix given by
\begin{equation}
\tilde{R}_{y_1,\cdots,y_k|x}  = \prod_{j=1}^k \tilde{Q}^j_{y_j|x}.
\end{equation}
\end{definition}
By definition, the product corresponds to simultaneously releasing $Q^1(x),\cdots,Q^k(x)$.
A key fact of LDP \cite{dwork2006calibrating} is that the leakage of the product is at most the sum of the leakages, i.e.

\begin{equation}
\recht{LDP}(Q^1 \times \cdots \times Q^k) \leq \sum_j \recht{LDP}(Q^j).
\end{equation}

In other words, when using different protocols simultaneously, one only has a fixed amount of `admissible leakage' that has to be distributed among the different protocols; this concept is known as a \emph{privacy budget}. A similar statement holds for the privacy metric $\recht{S}_\mu$; because this metric looks at privacy rather than leakage, we have to take the complement $1-\recht{S}_\mu$ to make a meaningful statement.

\begin{proposition} \label{prop:budget}
Let $Q^1,\cdots,Q^k$ be $k$ privacy protocols on $\mathcal{A}$. Then
\begin{equation}
    1-\recht{S}_{\mu}\left(Q^1 \times \cdots \times Q^k\right) \leq \sum_j \left(1-\recht{S}_\mu(Q^j)\right).
\end{equation}
\end{proposition}

\subsection{Mixture}

Apart from postcomposing privacy protocols or applying them in parallel, another way to combine multiple privacy protocols into one is by creating a \emph{mixture protocol}. Given privacy protocols $Q^1,\cdots, Q^k$ and a weight vector $w \in \mathcal{P}_{\{1,\cdots,k\}}$, their mixture $M_{w}(Q^1,\cdots,Q^k)$ is defined to be the privacy protocol that, for an input $x$, first draws an index $j\in\{1,\ldots,k\}$ according to the probability distribution $w$, and then outputs the pair $(j,Q^j(x))$.
\begin{definition}
Let $\mathcal{A}$ be a finite set, and let $Q^1 = (\tilde{Q}^1,\mathcal{B}^1),\cdots,Q^k = (\tilde{Q}^k,\mathcal{B}^k)$ be privacy protocols on $\mathcal{A}$. Let $w \in \mathcal{P}_{\{1,\cdots,k\}}$. Then the \emph{$w$-mixture} of $(Q^1,\cdots, Q^k)$ is the privacy protocol $M_{w}(Q^1,\cdots,Q^k) = (\tilde{R},\mathcal{C})$ with
\begin{equation}
    \tilde{R}_{(j,y)|x} = w_j  \tilde{Q}^j_{y|x},
    \quad\quad
    \mathcal{C} = \bigsqcup_j \mathcal{B}^j = \left\{(j,y) : j \leq k, y \in \mathcal{B}^j\right\}.
\end{equation}
\end{definition}
One can view a mixture protocol as the local equivalent of what has been called \emph{parallel composition} for DP protocols \cite{nguyen2013differential}. A key property of mixtures is that privacy behaves linearly under privacy, while asymptotic utility behaves superlinearly:

\begin{proposition} \label{prop:mixture}
Let $Q^1,\cdots,Q^k$ be faithful privacy protocols on $\mathcal{A}$, and let $w \in \mathcal{P}_{\{1,\cdots,k\}}$. Then
\begin{equation}
\recht{S}_{\mu}\left(M_{w}(Q^1,\cdots,Q^k)\right) = \sum_j w_j \recht{S}_{\mu}(Q^j). \label{eq:mixture_1}
\end{equation}
If $\mu$ is continuous we furthermore have
\begin{equation}
\recht{U}^{\recht{as}}_{\mu}\left(M_{w}(Q^1,\cdots,Q^k)\right) \geq \sum_j w_j \recht{U}^{\recht{as}}_{\mu}(Q^j). \label{eq:mixture_2}
\end{equation}
\end{proposition}

\begin{remark}
If we take privacy protocols of the same privacy level, then the Theorem above tells us that their mixture will also have the same privacy. The utility will at least be a weighed average of their utilities, and in fact it can even supersede the maximum of the asymptotic utilities. As an example, let $\mathcal{A} = \{1,2,3\}$, and let $Q^1$ and $Q^2$ be the privacy protocols given by the matrices
\begin{equation}
    \tilde{Q}^1 := \left(\begin{array}{ccc} 1 & 0 & 0 \\ 0 & \frac{2}{3} & \frac{1}{3} \\ 0 & \frac{1}{3} & \frac{2}{3}\end{array}\right), \ \ \ \ \tilde{Q}^2 := \left(\begin{array}{ccc} \frac{2}{3} & \frac{1}{3} & 0 \\ \frac{1}{3} & \frac{2}{3} & 0 \\ 0 & 0 & 1\end{array}\right).
\end{equation}
Let $\mu$ be the Dirichlet distribution with parameters $(1,1,1)$. Since this distribution is symmetric, and $Q^1$ and $Q^2$ differ only by a permutation of the input and output set, they will have the same privacy and asymptotic utility. In fact, using Def. \ref{def:uas} one can calculate that $\recht{U}^{\recht{as}}_{\mu}(Q^1) = \recht{U}^{\recht{as}}_{\mu}(Q^2) = -0.987$. However, $\recht{U}^{\recht{as}}_{\mu}(M_{(\frac{1}{2},\frac{1}{2})}(Q^1,Q^2)) = -0.691$. This shows that for $n$ sufficiently large, the mixture will outperform both $Q^1$ and $Q^2$, while offering the same level of privacy.
\end{remark}

 \section{The tradeoff between utility and privacy} 
 \label{sec:tradeoff}

An important aspect of the LDP setting is the \emph{privacy-utility tradeoff}. 
Intuitively, the more `random' a privacy protocol is, the less utility the aggregator will have, 
but the more the privacy of the user is maintained. 
Since one cannot have perfect utility without a total loss of privacy or vice versa,
we are necessarily in a situation where we have to pick a protocol that lies somewhere between these two extremes. 
Therefore, it is important to see to what extent this tradeoff can be characterised. 
If one takes the viewpoint that any information obtained by the aggregator constitutes utility, and any information withheld from the aggregator constitutes privacy, then clearly the tradeoff is very direct, and for any privacy protocol the utility and privacy sum up to a constant. 
For the new metrics introduced in this paper the tradeoff is far less direct, since privacy and utility regard different domains of information. 
Nevertheless, we are able to provide an upper bound on (asymptotic) utility in terms of (worst case) privacy:

\begin{theorem} \label{thm:tradeoff}
Let $Q$ be a faithful privacy protocol, and let $\mu$ be continuous. Then
\begin{eqnarray}
\recht{U}^{\recht{as}}_{\mu}(Q) &\leq& -\frac{1}{2}\recht{log}(2\pi \textrm{\emph{e}}) + \recht{log}\frac{1-\recht{S}^{\recht{wc}}(Q)}{\recht{S}^{\recht{wc}}(Q)}, \label{eq:tradeoff1}
\\
\recht{F}_{\mu}(Q) &\leq& \frac{\textrm{\emph{e}}^{-2\recht{C}_{\mu}}}{2\pi \textrm{\emph{e}}}\left(\frac{1-\recht{S}^{\recht{wc}}(Q)}{\recht{S}^{\recht{wc}}(Q)}\right)^2. \label{eq:tradeoff2}
\end{eqnarray}
\end{theorem}

Note that the two formulas above are straightforward consequences of each other. In general, the upper bound will not be strict if $a > 2$. The upper bound is more meaningful in the high privacy domain, as there $\recht{S}^{\recht{wc}}(Q) \approx 1$, and there it shows that $\recht{F}_{\mu}(Q)$ should be small. If we take the Jeffreys prior, one has $\recht{C}_{\mu} \rightarrow \infty$ as $a \rightarrow \infty$, which shows that for larger $a$ one generally has smaller values of $\recht{F}_{\mu}$ for the same level of worst-case privacy. However, $\recht{C}_{\mu}$ has no lower bound if we range over all possible $\mu$, so this Theorem does not yield an upper bound on $\recht{F}_{\mu}(Q)$ that does not depend on $\mu$.

\section{Posterior distributions}
\label{sec:posterior}

The utility measures introduced in this paper center on the mutual information between the outcome vector $\vec{Y}$ obtained by the aggregator and the aggregator's goal (either $P$ or $T$). This information manifests itself in the posterior distribution of this goal given the outcome vector~$\vec{y}$; since no information is lost in going from the received reports to the posterior distribution, the posterior is the natural tool for the aggregator to use the user reports. Therefore it is important for the aggregator to {\em efficiently} compute the posterior. In this section we discuss the (numerical) computation of the posterior for $P$. We suppose that $P$ is a smooth random variable, i.e. it has a probability density function $f$ as in (\ref{eq:de}).

\begin{proposition}
\label{prop:posterior1}
Let $\vec{y} \in \mathcal{B}^n$. For $\beta \in\cal B$ let $s_\beta$ be defined as
$s_\beta := \#\{i: y_i = \beta\}$.
The posterior distribution of $P$ given the observation $\vec{Y} = \vec{y}$ is
\begin{equation} \label{eq:posterior}
    f_{P|\vec{Y} = \vec{y}}(p)=\frac{f(p)\mathbb{P}(\vec{Y}=\vec{y}|P=p)}{C_{\vec{y}}}
\end{equation}
with
\begin{eqnarray}
    \mathbb{P}(\vec{Y}=\vec{y}|P=p) &=&
    \prod_{i=1}^n (\tilde Q p)_{y_i}
    =\prod_{\beta\in\cal B}[(\tilde Q p)_\beta]^{s_\beta(\vec{y})}
\label{eq:ygivenp}
    \\
    C_{\vec{y}} &=& \mathbb{P}(\vec{Y} = \vec{y}) = \int_{p \in \mathcal{P}_{\mathcal{A}}} f(p)\mathbb{P}(\vec{Y}=\vec{y}|P=p)\textrm{\emph{d}}p.
\label{eq:normc}
\end{eqnarray}
Let $t(\vec{x})$ be as in section \ref{ssec:ldp}.
If the prior $\mu$ is a Dirichlet distribution with parameter $\alpha$, then one can write
\begin{eqnarray}
f_{P|\vec{Y}=\vec{y}}(p) &=& \frac{1}{C_{\vec{y}}}\sum_{\vec{x} \in \mathcal{A}^n}\recht{B}(\alpha+t(\vec{x}))\left(\prod_i \tilde{Q}_{y_i|x_i}\right)\cdot \Delta_{\alpha+t(\vec{x})}(p), \label{eq:fpyp2} \\
    C_{\vec{y}} &=& \sum_{\vec{x}\in{\cal A}^n}\recht{B}(\alpha+t(\vec{x}))
    \prod_{i=1}^n \tilde{Q}_{y_i|x_i}. \label{eq:c2}
\end{eqnarray}
\end{proposition}

Note that this Proposition also covers the case where $\mu$ equals the Jeffreys prior. 
While (\ref{eq:fpyp2}) is a more complicated formula for $f_{P|\vec{Y} = \vec{y}}(p)$, its advantage is that we can get the marginal distributions of the $P_x$ directly from it, without doing any integration, as the marginal of a Dirichlet distribution is a beta distribution (i.e. a Dirichlet distribution in $2$ dimensions).

The $p$-dependence (\ref{eq:ygivenp}) of the posterior
is easy to evaluate, but the overall normalisation constant $C$ (\ref{eq:normc}) is computationally expensive. However, if the prior distribution is a Dirichlet distribution, there exist several efficient methods to obtain random samples \cite{cheng1998random,betancourt2012cruising}. This makes it feasible to perform a Monte Carlo integration in (\ref{eq:normc}). Furthermore, for established privacy protocols we can obtain less complicated formulas for the posterior (see Section \ref{sec:protocols}).

\section{Analysis of privacy protocols}
\label{sec:protocols}

In this section we discuss how various established  protocols behave with respect to the new metrics. 

We will set the prior $\mu$ to be the Jeffreys prior, which is the Dirichlet measure on $\mathcal{P}_A$ with parameters $\alpha_i=1/2$
for all $i\in\cal A$ (see section \ref{ssec:dirichlet}).
This choice yields the following result for $\recht{C}_{\mu}$ (defined in section \ref{ssec:lim}) and $\recht{H}(X_i|P)$:
\begin{eqnarray}
    \recht{C}_{\mu} &=& -\frac{1}{2}\recht{log}(2\pi\textrm{e})+\frac{a}{2a-2}\left(\dg\left(\frac{a}{2}\right)-\dg\left(\frac{1}{2}\right)\right),\\
    \recht{H}(X_i|P) &=& \dg\left(\frac{a+2}{2}\right)-\dg\left(\frac{3}{2}\right).
\end{eqnarray}
In section \ref{ssec:complexity} we will discuss the computational complexity of the metrics discussed in this paper, both for general protocols and for specific established protocols. 
It turns out that for concrete protocols, the calculation will often be easier than it is in the general case.

\subsection{Generalised Randomised Response}
\label{ssec:grr}

Let $\varepsilon \in \mathbb{R}_{>0}$. 
The \emph{Generalised Randomised Response} \cite{warner1965randomized} is the privacy protocol $\recht{GRR}^{a,\varepsilon} = (\tilde{Q},\mathcal{B})$ on $\mathcal{A}$ satisfying $\mathcal{B} = \mathcal{A}$ and
\begin{equation}
\tilde{Q}_{y|x} = \left\{\begin{array}{ll} \frac{\textrm{e}^{\varepsilon}}{\textrm{e}^{\varepsilon}+a-1},& \textrm{ if $y=x$,}\\
\frac{1}{\textrm{e}^{\varepsilon}+a-1},& \textrm{ if $y\neq x$.}\end{array}\right.
\end{equation}
By design this privacy protocol satisfies $\recht{LDP}(\recht{GRR}^{a,\varepsilon}) = \varepsilon$. In the Propositions below, we use the abbreviated notation $\beta := \textrm{e}^{\varepsilon}-1$.

\begin{proposition} \label{prop:grrutipriv}
Let $x \in \mathcal{A}$ and $i \leq n$ both be arbitrary. Then
\begin{eqnarray} 
    \recht{S}_{\mu}(\recht{GRR}^{a,\varepsilon}) &=& 1-\frac{\varepsilon\textrm{\emph{e}}^{\varepsilon}-a\mathbb{E}_{P_x}[(1+\beta P_x)\log(1+\beta P_x)]}{(a+\beta)\recht{H}(X_i|P)}, \label{eq:grruti}
    \\
    \recht{U}^{\recht{as}}_{\mu}(\recht{GRR}^{a,\varepsilon}) &=& -\frac{1}{2}\log (2\pi\textrm{\emph{e}}) + \log \beta-\frac{a-2}{2a-2}\log(a+\beta) - \frac{a}{2a-2} \mathbb{E}_{P_x}\log(1+\beta P_x).
\end{eqnarray}
\end{proposition}

Note that $P_x$ follows a beta distribution with parameters $(\tfrac{1}{2},\tfrac{a-1}{2})$, hence both the privacy and the asymptotic utility of $\recht{GRR}^{a,\varepsilon}$ are readily computed. We also find a more succinct description of the posterior distribution:

\begin{proposition} \label{prop:grrpost}
Let $\vec{y} \in \mathcal{A}^n$. 
The posterior distribution of $\bar{P}$ given $\vec{Y} = \vec{y}$ (or equivalently, given $S = s$) under $\recht{GRR}^{a,\varepsilon}$, and given the Jeffreys prior distribution, equals
\begin{equation} 
f_{P|\vec{Y} = \vec{y}}(p) = \frac{1}{C_{\vec{y}}}
\prod_{x\in\cal A} p_x^{-\frac{1}{2}}(1+\beta p_x)^{s_x}
\end{equation}
where
\begin{equation}
    C_{\vec{y}} = \sum_{\substack{k \in \mathbb{Z}_{\geq 0}^{\mathcal{A}}:\\ \forall x \ k_x \leq s_x}} \recht{B}(k+\alpha)\beta^{\sum_x k_x} \prod_x \binom{s_x}{k_x}.
\end{equation}
\end{proposition}

For GRR we can also find expressions for $\recht{I}(\vec{Y};\bar{P})$ that are more easily computed than the general case (see appendix \ref{sapp:grr}). We can use this to compute $\recht{U}^{\recht{distr}}_{n,\mu}(Q)$ and $\recht{U}^{\recht{digit}}_{n,\mu}(Q)$.

\subsection{Unary Encoding}

Let $\kappa, \lambda \in [0,1]$ with $\kappa \geq \lambda$. Unary Encoding \cite{wang2017locally} with parameters $(\kappa,\lambda)$ is defined to be the privacy protocol $\recht{UE}^{a,\kappa,\lambda} = (\tilde{Q}^{a,\kappa,\lambda},2^{\mathcal{A}})$ obtained as follows: Given an input $x \in \mathcal{A}$, the output set $\recht{UE}^{a,\kappa,\lambda}(x)$ contains $x$ with probability $\kappa$, and each of the other values with probability $\lambda$, and all these probabilities are independent. As such, the associated matrix coefficient for $x \in \mathcal{A}, y \subset \mathcal{A}$ is

\begin{eqnarray}
    \tilde{Q}^{a,\kappa,\lambda}_{y|x} &=& \left\{\begin{array}{ll}
    \kappa\cdot \lambda^{\#y-1} \cdot (1-\lambda)^{a-\#y}, & \textrm{ if $x \in y$,}\\
    (1-\kappa) \cdot \lambda^{\#y} \cdot (1-\lambda)^{a-\#y-1}, & \textrm{ if $x \notin y$;}
    \end{array}\right.\\
    &=& \lambda^{\#y-1}\cdot(1-\lambda)^{a-\#y-1} \cdot \left[\lambda(1-\kappa) + \delta_{x \in y}(\kappa-\lambda)\right].
\end{eqnarray}

Note that $\recht{LDP}(\recht{UE}^{a,\kappa,\lambda}) = \recht{log}\frac{\kappa(1-\lambda)}{\lambda(1-\kappa)}$. At LDP $\varepsilon$, some common choices of $(\kappa,\lambda)$ are listed in the Table \ref{tab:ue}.

\begin{table}[h] 
\centering
\begin{tabular}{lcc} 
 & $\kappa$ & $\lambda$ \\
 \hline
 Basic RAPPOR \cite{erlingsson2014rappor} & $\frac{\textrm{e}^{\varepsilon/2}}{\textrm{e}^{\varepsilon/2}+1}$ & $\frac{1}{\textrm{e}^{\varepsilon/2}+1}$ \\
 OUE \cite{wang2017locally} & $\frac{1}{2}$ & $\frac{1}{\textrm{e}^{\varepsilon}+1}$ \\
 BLH \cite{bassily2015local,wang2017locally} & $\frac{\textrm{e}^{\varepsilon}}{\textrm{e}^{\varepsilon}+1}$ & $\frac{1}{2}$\\
 \hline
\end{tabular}
\caption{Three commonly used versions of UE.}
\label{tab:ue}
\end{table}

For the privacy of UE we can derive a reasonably straightforward expression:

\begin{proposition} \label{prop:uepriv}
For $0 \leq g \leq a$, let $R_s$ be the continuous random variable of the form $R_g = \lambda^{g-1}(1-\lambda)^{a-g-1}(\lambda(1-\kappa)+(\kappa-\lambda)B_g)$, where $B_g$ is drawn from a beta distribution with parameters $(\tfrac{g}{2},\tfrac{a-g}{2})$; if $g = 0$ or $g = a$, let $B_g$ be constant $0$ or $1$, respectively. Let $\recht{H}_{\recht{b}}$ be the binary entropy function (see section \ref{ssec:it}). Then
\begin{equation}
\recht{S}_{\mu}(\recht{UE}^{a,\kappa,\lambda}) = 1-\frac{-(a-1)\recht{H}_{\recht{b}}(\lambda) - \recht{H}_{\recht{b}}(\kappa)-\sum_{g=0}^a \binom{a}{g} \mathbb{E}_{R_g}\left[R_g \log R_g\right]}{\recht{H}(X_i|P)}.
\end{equation}
\end{proposition}

Unfortunately, we do not know of the existence of a `simple' formula for any of the utility measures introduced in this paper. Nevertheless, one can calculate $\recht{I}(\vec{Y};P)$ (and with that the digit and distribution utility) more efficiently than in the general case, see appendix \ref{sapp:ue}.

\subsection{Overview of computation complexity}
\label{ssec:complexity}

In this section we give an overview of how computationally involved the various metrics are. In Table \ref{tab:comp}, $\langle x,y \rangle$ means that to compute the value of the metric, we need to compute $x$ different $y$-dimensional integrals ($y = 0$ means that we need to compute $x$ constants). The $\mathcal{O}$-symbol refers to the behaviour in $n$ as $n \rightarrow \infty$. As one can see, for UE, and especially for GRR, we can significantly improve the computational complexity of the posterior of the marginals and $\recht{S}_{\mu}$. With regards to $\recht{U}^{\recht{distr}}_{n,\mu}$ and $\recht{U}^{\recht{digit}}_{n,\mu}$, we have found a way to increase the number of summands, but the summands no longer contain integrals. In practice, this makes computing these metrics easier.

\begin{table}[h]
\centering
\begin{tabular}{lccc}
     & general & GRR ($b = a$) & UE ($b = 2^a$) \\
     \hline
    $\recht{U}^{\recht{distr}}_{n,\mu}$ & $\langle \mathcal{O}(n^{b-1}),a\rangle$ & $\langle\mathcal{O}(n^{2a-1}),0\rangle + \langle 1,1\rangle$ & $\langle \mathcal{O}(n^{1+(a+1)2^{a-1}}),0 \rangle + \langle \mathcal{O}(1),1 \rangle$\\
    $\recht{U}^{\recht{tally}}_{n,\mu}$ & $\langle\mathcal{O}(n^{a-1}),a \rangle + \langle \mathcal{O}(n^{ab-1}),0 \rangle$ & = gen. & = gen. \\
    $\recht{S}_{\mu}$ & $\langle \mathcal{O}(1), a\rangle$ & $\langle \mathcal{O}(1),1\rangle$ & $\langle \mathcal{O}(1),1 \rangle$ \\
    \hline
    $\recht{U}^{\recht{digit}}_{n,\mu}$ & $\langle \mathcal{O}(n^{b-1}),a\rangle$ & $\langle\mathcal{O}(n^{2a-1}),0\rangle + \langle 1,1\rangle$ & $\langle \mathcal{O}(n^{1+(a+1)2^{a-1}}),0 \rangle + \langle \mathcal{O}(1),1 \rangle$\\
    $\recht{U}^{\recht{as}}_{\mu}$ & hard\footnote{$a$-dimensional integral of a logarithm of a sum of $(a-1)!b^{a-1}$ terms.} & easy\footnote{1-dimensional integral of a logarithm of $1$ term.} & =gen. \\
    \hline
    posterior\footnote{Assuming a Dirichlet prior.} & $\langle \mathcal{O}(n^{b(a-1)},0\rangle$ & $\langle \mathcal{O}(n^{a}),0\rangle$ & = gen.\\
\end{tabular} 
\caption{Computational complexity of the metrics.}
\label{tab:comp}
\end{table}

\subsection{Numerical comparisons}

We use the new metrics to compare different privacy protocols. In Figure \ref{fig:priv} the average privacy, for a fixed level of worst-case privacy, is plotted against $a$ for different privacy protocols. As one can see, there is a big difference between average and worst-case privacy, that only grows with $a$. For large $a$, GRR offers more privacy on average for the same level of worst-case privacy than the various versions of UE.

\begin{figure}
    \centering
    \includegraphics[width = 0.5\linewidth]{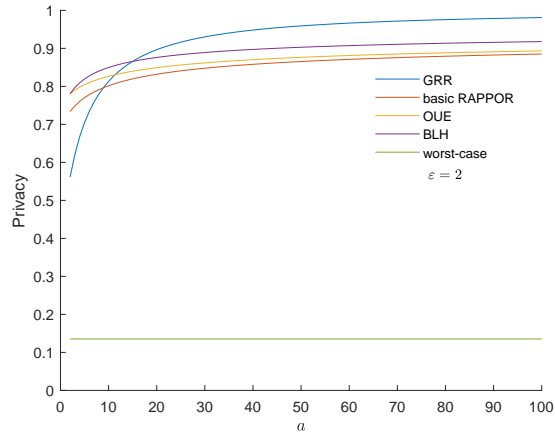}
    \caption{Privacy vs worst-case privacy for GRR and UE, for varying $a$  ($\varepsilon = 2$).}
    \label{fig:priv}
\end{figure}

In Figure \ref{fig:tradeoff} we plot the tradeoff between (asymptotic) utility and (average and worst-case) privacy for GRR and UE, by varying the parameter $\varepsilon$. Note that all protocols have $\recht{S}^{\recht{wc}} \rightarrow 0$ as $\varepsilon \rightarrow \infty$. 
In this limit case GRR and basic RAPPOR also satisfy $\recht{S}_{\mu} \rightarrow 0$ and $\recht{F}_{\mu} \rightarrow 1$.
On the other hand, for OUE and BLH both $\recht{S}_{\mu}$ and $\recht{F}_{\mu}$ converge to nontrivial values as $\varepsilon \rightarrow 1$.
As $\varepsilon \rightarrow 0$, all protocols satisfy $\recht{S}^{\recht{wc}},\recht{S}_{\mu} \rightarrow 1$ and $\recht{F}_{\mu} \rightarrow 0$.
Interestingly, the best protocol to use depends on one's privacy metric: For a given level of average privacy, the best utility is offered either by OUE (in the high privacy domain) or basic RAPPOR (in the low privacy domain). For a given level of worst-case privacy, the best utility is offered by GRR.

\begin{figure}
    \centering
    \subfloat[]{\begin{minipage}{0.45\linewidth}
    \includegraphics[width=0.9\linewidth]{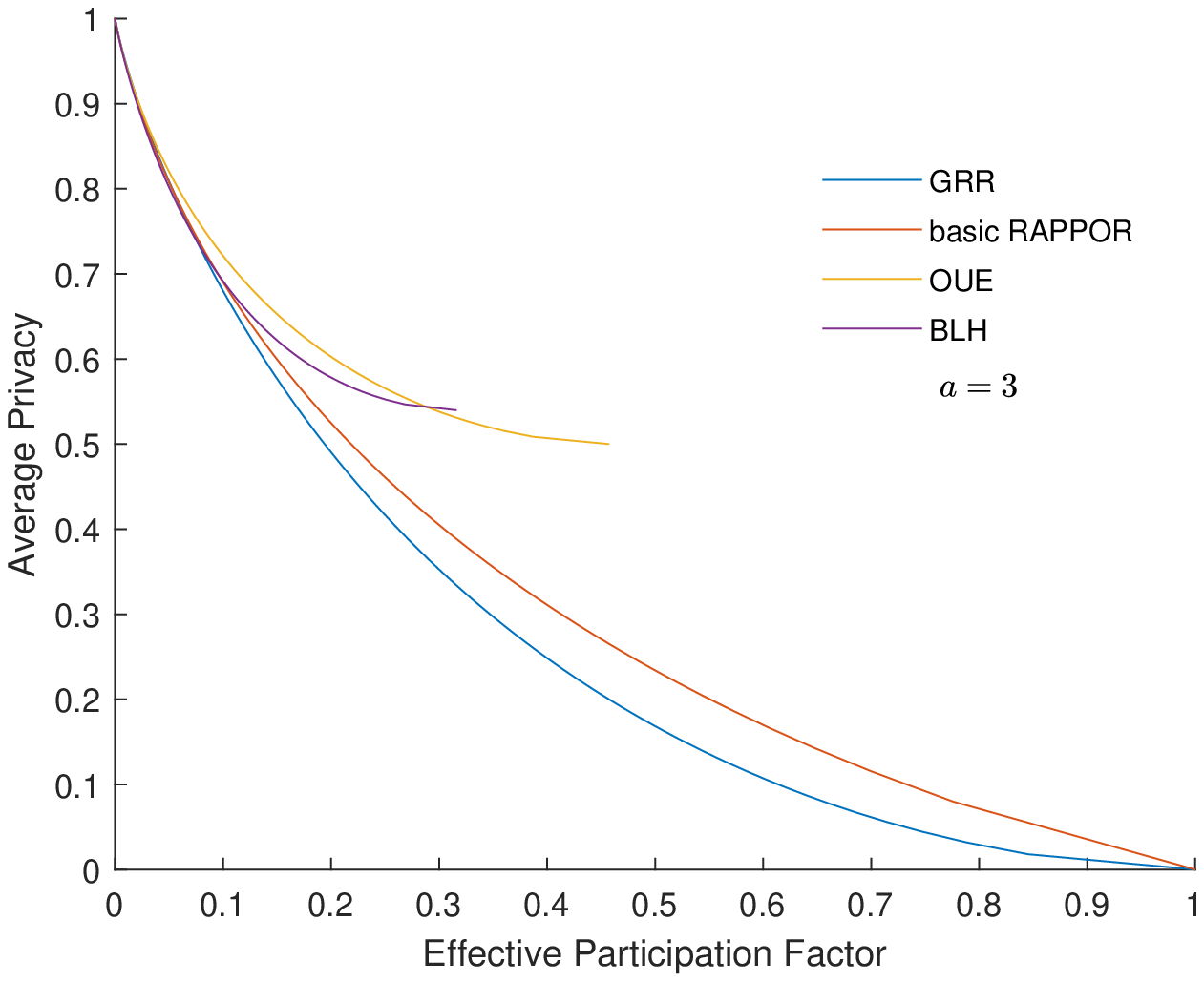}   
    \end{minipage}} \hfill
    \subfloat[]{\begin{minipage}{0.45\linewidth}
    \includegraphics[width=0.9\linewidth]{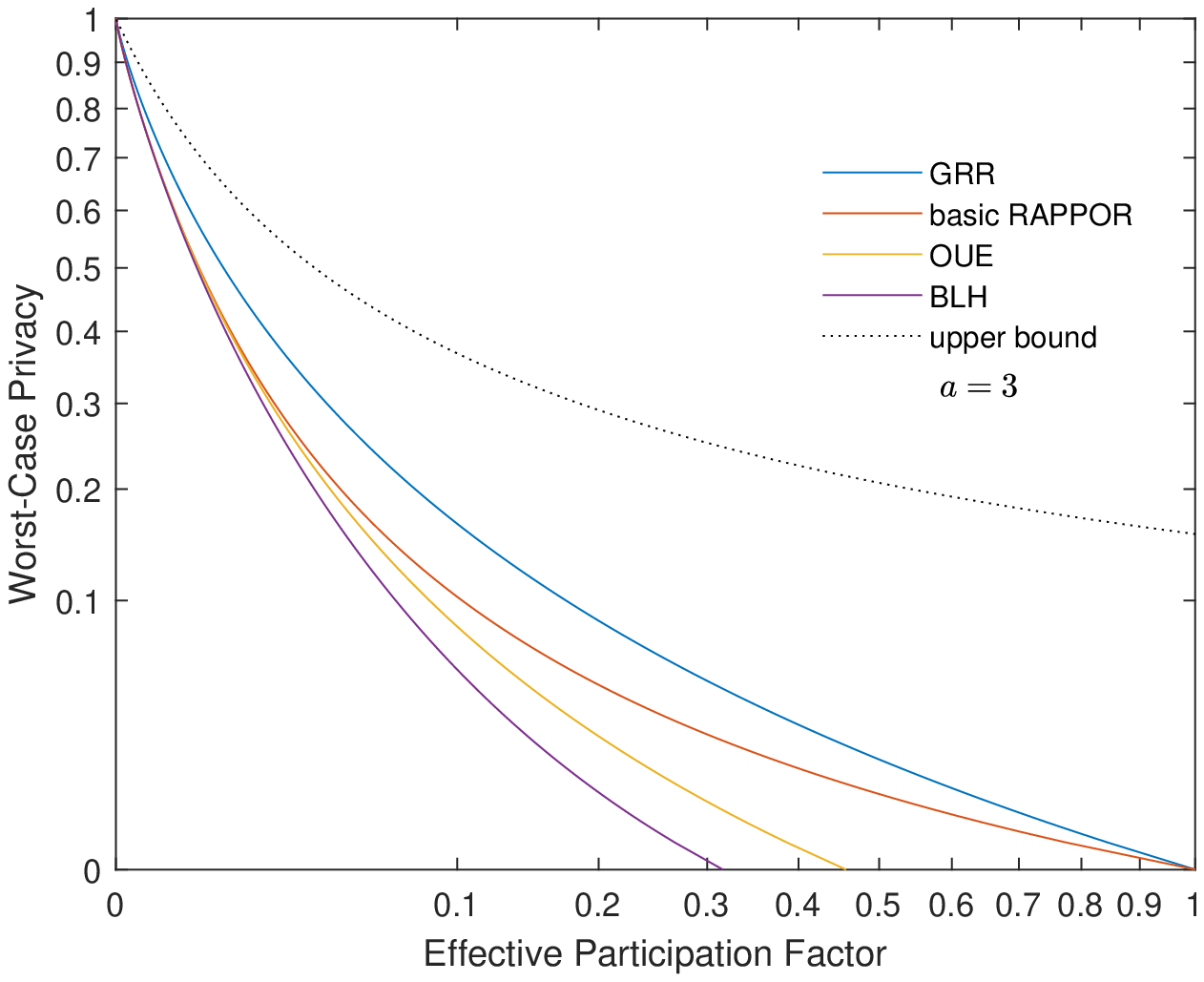}
    \end{minipage}}
    \caption{Privacy-utility tradeoff for GRR and UE by varying $\varepsilon$ between $0$ and $\infty$ ($a = 3)$. The axes in (b) are scaled by their square root for legibility.}
    \label{fig:tradeoff}
\end{figure}

\section{Future work}
\label{sec:conclusion}

This paper suggests several directions for future research. On the theoretical side, further research is needed into formal aspects of the privacy-utility tradeoff. While Theorem \ref{thm:tradeoff} characterises the tradeoff in terms of worst-case privacy, for a more complete understanding of the relation between privacy and security we would like to have a theoretical result that establishes the relation between $\recht{S}_{\mu}(Q)$ and any of the utility metrics. Furthermore, apart from such a theoretical approach, we would like to have a way to find optimal privacy protocols, i.e. for a given prior distribution and number of users and a chosen level of privacy $L$, find the privacy protocol $Q$ that maximises $\recht{U}^{\recht{distr}}_{n,\mu}(Q)$ while satisfying $\recht{S}_{\mu}(Q) \leq L$. 

On the more practical side, the biggest disadvantage of the new metrics is that they can be computationally complex. Although the results in section \ref{ssec:complexity} show that for GRR and UE we are able to find faster ways to compute metrics than the general case, we have not found a way to do this for all metrics, or for general privacy protocols. Ultimately, we would want to find approximations to the utility and privacy metrics that are easily computed, and for which we can prove that they bear a close relation to the actual metrics.

\DeclareRobustCommand{\VAN}[3]{#3}

\appendix

\section{Short Proofs}

In this appendix, we provide proofs for the mathematical statements in the main text, except for those in Sections \ref{sec:wc}, \ref{sec:limit}, and \ref{sec:protocols}. The proofs of those sections require more preliminary work and are treated in the other appendices.

\begin{proof}[Proof of Lemma \ref{lem:priv}]
Given $P$, the $X_i$ are independent and all have the same probability distribution, hence
\begin{equation}
    \frac{\recht{H}(\vec{X}|\vec{Y},{P})}{\recht{H}(\vec{X}|{P})} = \frac{\sum_i \recht{H}(X_i|Y_i,{P})}{\sum_i \recht{H}(X_i|{P})} = \frac{n \cdot  \recht{H}(X_1|Y_1,{P})}{n \cdot \recht{H}(X_1|{P})} = \frac{\recht{H}(X_1|Y_1,{P})}{\recht{H}(X_1|{P})},
\end{equation}
of which the right hand side does not depend on $n$.
\end{proof}

\begin{proof}[Proof of Lemma \ref{lem:digit}]
Since ${P}$ is a continuous random variable on the $(a-1)$-dimensional affine space $\mathcal{P}_{\mathcal{A}}$, one has that $\recht{H}({P}_{\langle d \rangle}) \approx \recht{h}({P}) + (a-1)d$. The Lemma follows from substituting this in (\ref{eq:udigit}).
\end{proof}

\begin{proof}[Proof of Proposition \ref{prop:faithful}.]
First suppose $Q$ is faithful. By definition this means that $\recht{rk}(\tilde{Q}) = a$, hence $\tilde{Q}$, as a map $\mathbb{R}^a \rightarrow \mathbb{R}^b$, is injective. In general, $Q_*$ acts on $\mathcal{P}_{\mathcal{A}}$ by $p \mapsto \tilde{Q}\cdot p$. Since $\tilde{Q}$ is injective, this means that $Q_*$ is injective.

Now suppose $Q$ is not faithful, and let $v \in \mathbb{R}^{\mathcal{A}}\setminus\{0\}$ be such that $\tilde{Q}v = 0$. Then
\begin{equation}
0 = \sum_y (\tilde{Q}v)_y = \sum_y \sum_x \tilde{Q}_{y|x} v_x = \sum_x v_x.
\end{equation}
Let $p \in \mathcal{P}_{\mathcal{A}}$ be such that all $p_x$ are positive. Then for all $\lambda \in \mathbb{R}_{>0}$ one has $\sum_x (p + \lambda v)_x = 1$. By taking $\lambda \in \mathbb{R}_{>0}$ small enough, one can also insure that $(p+\lambda v)_x \geq 0$ for all $x$. Then $p+\lambda v \in \mathcal{P}_{\mathcal{A}}$, but by definition one has $Q_*(p+\lambda v) = \tilde{Q} \cdot (p+\lambda v) = \tilde{Q} \cdot p = Q_*p$. This shows that $Q_*$ is not injective.
\end{proof}

\begin{proof}[Proof of Proposition \ref{prop:postprocessing}]
Write $Z_i = R(Y_i) = (R \circ Q)(X_i)$. The $Y_i$ are each drawn independently from $Q_*{P}$, which is a random variable on $\mathcal{P}_{{B}}$ drawn from the probability distribution $Q_{**}\mu$ defined in Def. \ref{def:pushforward}; as such one has $\recht{U}^{\recht{distr}}_{n,Q_{**}\mu}(R) = \frac{\recht{I}(\vec{Z};Q_*{P})}{\recht{I}(\vec{Y};Q_*{P})}$. Furthermore, we have two Markov chains ${P} \rightarrow \vec{X} \rightarrow \vec{Y} \rightarrow \vec{Z}$ and ${P} \rightarrow Q_*{P} \rightarrow \vec{Y} \rightarrow \vec{Z}$. The first one implies $\recht{I}(\vec{X};{P}) \geq \recht{I}(\vec{Y};{P}) \geq \recht{I}(\vec{Z};{P})$. Since $Q_*$ is a deterministic function, the second Markov chain shows that $\recht{I}(\vec{Y};{P}) = \recht{I}(\vec{Y};Q_*{P})$ and $\recht{I}(\vec{Z};{P}) = \recht{I}(\vec{Z};Q_*{P})$. Putting this together, we get
\begin{eqnarray}
\recht{U}^{\recht{distr}}_{n,\mu}(R \circ Q) &=& \frac{\recht{I}(\vec{Z};{P})}{\recht{I}(\vec{X};{P})} \leq \frac{\recht{I}(\vec{Y};{P})}{\recht{I}(\vec{X};{P})} = \recht{U}^{\recht{distr}}_{n,\mu}(Q);\\
\recht{U}^{\recht{distr}}_{n,\mu}(R \circ Q) &=& \frac{\recht{I}(\vec{Z};{P})}{\recht{I}(\vec{X};{P})} \leq \frac{\recht{I}(\vec{Z};{P})}{\recht{I}(\vec{Y};{P})} = \frac{\recht{I}(\vec{Z};Q_*{P})}{\recht{I}(\vec{Y};Q_*{P})} = \recht{U}^{\recht{distr}}_{n,Q_{**}\mu}(R);\\
\recht{U}^{\recht{digit}}_{n,\mu}(R \circ Q) &=& \frac{\recht{I}(\vec{Z};P)-\recht{h}(P)}{a-1} \leq \frac{\recht{I}(\vec{Y};P)-\recht{h}(P)}{a-1} = \recht{U}^{\recht{digit}}_{n,\mu}(Q).
\end{eqnarray}
This proves (\ref{eq:postprocessing_1}) and (\ref{eq:postprocessing_5}). Analogous to the above we also get $\recht{I}(\vec{Z};{T}) \leq \recht{I}(\vec{Y};{T})$, hence
\begin{equation}
\recht{U}^{\recht{tally}}_{n,\mu}(R \circ Q) = \frac{\recht{I}(\vec{Z};{T})}{\recht{H}({T})} \leq \frac{\recht{I}(\vec{Y};{T})}{\recht{H}({T})} = \recht{U}^{\recht{tally}}_{n,\mu}(Q),
\end{equation}
which proves (\ref{eq:postprocessing_2}). For any individual $i$ we have the same Markov chain ${P} \rightarrow X_i \rightarrow Y_i \rightarrow Z_i$, hence equation (\ref{eq:postprocessing_3}) is proven by
\begin{equation}
\recht{S}_{\mu}(R \circ Q) = \frac{\recht{H}(X_i|Z_i;{P})}{\recht{H}(X_i|{P})} \geq \frac{\recht{H}(X_i|Y_i;{P})}{\recht{H}(X_i|{P})} = \recht{S}_{\mu}(Q).
\end{equation}
It is a standard result for (local) differential privacy that $\recht{LDP}(R \circ Q) \leq \recht{LDP}(Q)$ \cite{dwork2006calibrating}. To prove (\ref{eq:postprocessing_4}), note that it follows from Theorem \ref{thm:wc}:
\begin{equation}
    \recht{S}^{\recht{wc}}(R \circ Q) = \textrm{e}^{-\recht{LDP}(R \circ Q)} \geq \textrm{e}^{-\recht{LDP}(Q)} = \recht{S}^{\recht{wc}}(Q). \qedhere
\end{equation}
\end{proof}

\begin{proof}[Proof of Proposition \ref{prop:budget}]
For simplicity, we prove this only for $k=2$, as the general case easily follows by induction. In this case we write $Y_i^1 := Q^1(X_i)$, $Y_i^2 := Q^2(X_i)$.
Since $Y_i^1$ and $Y_i^2$ are independent given $X_i$, we get that $\recht{H}(Y_i^1,Y_i^2|X_i) = \recht{H}(Y_i^1|X_i)+\recht{H}(Y_i^2|X_i)$; hence
\begin{equation}
\recht{I}(Y_i^1,Y_i^2;X_i|P) - \recht{I}(Y_i^1;X_i|P)-\recht{I}(Y_i^2;X_i|P) = \recht{H}(Y_i^1,Y_i^2|P) - \recht{H}(Y_i^1|P)-\recht{H}(Y_i^2|P) \leq 0.
\end{equation}
It follows that
\begin{equation}
1-\recht{S}_{\mu}(Q^1 \times Q^2) = \frac{\recht{I}(Y_i^1,Y_i^2;X_i|P)}{\recht{H}(X_i|P)} \leq \frac{\recht{I}(Y_i^1;X_i|P)+\recht{H}(Y_i^2;X_i|P)}{\recht{H}(X_i|P)} = (1-\recht{S}_\mu(Q^1))+(1-\recht{S}_\mu(Q^2)). \qedhere
\end{equation}
\end{proof}

\begin{proof}[Proof of Proposition \ref{prop:mixture}]
Let $J \in \{1,\cdots,k\}$ be random with probability vector $w$, and set $R:= M_{{w}}(Q^1,\cdots,Q^k)$. then
$\recht{H}(X_i|R(X_i),P) = \mathbb{E}_{J}\left[\recht{H}(X_i|Q^J(X_i),P)\right]$.
To prove (\ref{eq:mixture_1}), observe that from this it follows that
\begin{equation}
 \recht{S}_\mu(M_{{w}}(Q_1,\cdots,Q_k)) = \frac{\recht{H}(X_i|R(X_i),P)}{\recht{H}(X_i|P)} = \sum_j w_j \frac{\recht{H}(X_i|Q^j(X_i),P)}{\recht{H}(X_i|P)} = \sum_j w_j\recht{S}_\mu(Q^j).
\end{equation}
Now let $p \in \mathcal{P}_{\mathcal{A}}$. Write $M_w(Q_1,\cdots,Q_k) = (\tilde{R},\bigsqcup_j \mathcal{B}^j)$. Then for any $x,x'\leq a$ one has
\begin{equation}
(\tilde{R}^{\recht{T}}D_p\tilde{R})_{x,x'} = \sum_{(j,y) \in \bigsqcup_j \mathcal{B}^j} \frac{w_j^2 \tilde{Q}^j_{y|x}\tilde{Q}^j_{y|x'}}{\sum_{x''}w_j\tilde{Q}^j_{y|x''}}
= \sum_j w_j \sum_{y \in \mathcal{B}^j} \frac{\tilde{Q}^j_{y|x}\tilde{Q}^j_{y|x'}}{\sum_{x''}\tilde{Q}^j_{y|x''}}
= \sum_j w_j ((\tilde{Q}^j)^{\recht{T}}D_p\tilde{Q}^j)_{x,x'}.
\end{equation}
Since $\recht{log} \recht{det}$ is concave on positive definite symmetric matrices, we find that
\begin{equation}
    \recht{log}(\recht{det}(\tilde{R}^{\recht{T}}D_p\tilde{R}) \geq \sum_j w_j \recht{log}(\recht{det}((\tilde{Q}^j)^{\recht{T}}D_p\tilde{Q}^j).
\end{equation}
Applying this to Def. \ref{def:uas} gives us (\ref{eq:mixture_2}).
\end{proof}

\begin{proof}[Proof of Theorem \ref{thm:tradeoff}]
For this proof, it is more convenient to use the description of $\recht{U}^{\recht{as}}_{\mu}(Q)$ from Lemma \ref{lem:uasexp}, and the definition of $J_{Q_*p}$ given in its proof, i.e. its diagonal entries are given by $(J_{Q_*p})_{x,x} = \sum_{y=1}^b \frac{(\tilde{Q}_{y|x}-\tilde{Q}_{y|a})^2}{(Q_*p)_y}$. Let $y \leq b$ and $x < a$; then one has $|\tilde{Q}_{y|x}-\tilde{Q}_{y|a}| \leq (\textrm{e}^{\recht{LDP}(Q)}-1)^2\recht{min}_{x'}\tilde{Q}_{y|x'}$ and $(Q_*p)_y \geq \recht{min}_{x'}\tilde{Q}_{y|x'}$. It follows that
\begin{equation}
(J_{Q_*p})_{x,x} \leq \sum_y\frac{(\textrm{e}^{\textrm{LDP}(Q)}-1)^2(\recht{min}_x \tilde{Q}_{y|x})^2}{m_y} \leq \sum_y(\textrm{e}^{\textrm{LDP}(Q)}-1)^2 m_y = (\textrm{e}^{\textrm{LDP}(Q)}-1)^2.
\end{equation}
By Theorem \ref{thm:wc}, the right hand side is equal to $\frac{(1-\recht{S}^{\recht{wc}}(Q))^2}{(\recht{S}^{\recht{wc}}(Q))^2}$.
By Hadamard's inequality we know that the determinant of a positive definite matrix is at most the product of its diagonal elements, hence $\recht{log}(\recht{det}(J_{Q_*p})) \leq (2a-2)\recht{log}\frac{1-\recht{S}^{\recht{wc}}(Q)}{\recht{S}^{\recht{wc}}(Q)}$. If we plug this in into Lemma \ref{lem:uasexp} we find (\ref{eq:tradeoff1}). We get (\ref{eq:tradeoff2}) directly from substituting this into Def. \ref{def:participation}.
\end{proof}

\begin{proof}[Proof of Proposition \ref{prop:posterior1}]
For $p \in \mathcal{P}_{\mathcal{A}}$ one has $f_{P|\vec{Y}=\vec{y}}(p) = \frac{f(p)\mathbb{P}(\vec{Y} = \vec{y}|P=p)}{\mathbb{P}(\vec{Y}=\vec{y})}$. Given $P = p$, each $Y_i$ is drawn from $\mathcal{B}$ independently with probability distribution $\tilde{Q}p$; this proves (\ref{eq:ygivenp}), and (\ref{eq:normc}) follows directly from this. Now suppose $\mu$ is a Dirichlet distribution with parameter $\alpha$, then $f_P(p) = \recht{B}(\alpha)^{-1}\prod_{\gamma \in \mathcal{A}} p_{\gamma}^{\alpha_{\gamma}-1}$ and 
\begin{equation}
\mathbb{P}(\vec{Y} = \vec{y} | {P} = {p}) = \sum_{\vec{x} \in \mathcal{A}^n} \mathbb{P}(\vec{Y} = \vec{y}|\vec{X} = \vec{x})\mathbb{P}(\vec{X} = \vec{x}|{P} = {p}) = \sum_{\vec{x} \in \mathcal{A}^n} \\prod_i \tilde{Q}_{y_i|x_i}\cdot \prod_{\gamma \in \mathcal{A}} p_{\gamma}^{t_{\gamma}(\vec{x})}.
\end{equation}
Applying this to (\ref{eq:posterior}), we find
\begin{equation}
f_{P|\vec{Y}=\vec{y}}(p) = \frac{1}{\mathbb{P}(\vec{Y} = \vec{y}) \cdot \recht{B}(\alpha)} \sum_{\vec{x}}\prod_i \tilde{Q}_{y_i|x_i} \prod_{\gamma \in \mathcal{A}} p_{\gamma}^{\alpha_{\gamma}+t_{\gamma}(\vec{x})-1}
\end{equation}
which shows that the posterior distribution is proportional to
\begin{equation}
\sum_{\vec{x}}\recht{B}(\alpha+t(\vec{x}))\prod_i \tilde{Q}_{y_i|x_i} \cdot \Delta_{\alpha+t(\vec{x})}(p).
\end{equation}
We get the right constant from the fact that each Dirichlet distribution is a probability distribution, hence its integral over all $\mathcal{P}_{\mathcal{A}}$ equals $1$.
\end{proof}

\section{Proofs for worst-case privacy} \label{app:wc}

The main goal of this section is to prove Theorem \ref{thm:wc}. First we present the proof of the other statement of section \ref{sec:wc}.

\begin{proof}[Proof of Lemma \ref{lem:wc}]
Let $\mu$ be a prior distribution on $\mathcal{P}_{\mathcal{A}}$. Then
\begin{equation}
\recht{S}^{\recht{wc}}(Q) \leq \recht{inf}_{y \in \mathcal{B}} \frac{\recht{H}(X_i|Y = y,{P})}{\recht{H}(X_i|{P})} \leq \mathbb{E}_y \left[\frac{\recht{H}(X_i|Y = y,{P})}{\recht{H}(X_i|{P})}\right] = \recht{S}_\mu(Q). \qedhere
\end{equation}
\end{proof}

Before we give the proof of the main Theorem, we first need two auxiliary Lemmas.

\begin{lemma} \label{lem:avg}
Let $\mathcal{X}$ be a measure space with measure $\chi$, and let $f,g\colon \mathcal{X} \rightarrow \mathbb{R}_{\geq 0}$ be two functions with $g(x) >0$ for all $x$ and $0 < \int_{x \in \mathcal{X}} g(x)\textrm{d}\chi < \infty$. Then
\begin{equation}
    \frac{\int_{x \in \mathcal{X}} f(x)\textrm{d}\chi}{\int_{x \in \mathcal{X}} g(x)\textrm{d}\chi} \geq \recht{inf}_{x \in \mathcal{X}} \frac{f(x)}{g(x)}. \label{eq:difprivinf}
\end{equation}
\end{lemma}

\begin{proof}
Let $\nu$ be the probability measure on $\mathcal{X}$ given by $\nu(U) = \frac{1}{\int_{x \in \mathcal{X}} g(x) \textrm{d}\chi}\int_{x \in U}g(x)\textrm{d}\chi$. Then the left hand side of (\ref{eq:difprivinf}) equals $\int_{x \in \mathcal{X}} \frac{f(x)}{g(x)}\textrm{d}\nu = \mathbb{E}_{X \sim \nu}\left[\tfrac{f(X)}{g(X)}\right]$, which proves the Lemma.
\end{proof}

\begin{lemma} \label{lem:fracappr}
Let $t \in (0,1)$ and $A,B \in \mathbb{R}_{\geq 0}$ with $(A,B) \neq (0,0)$. Then
\begin{equation}
    \frac{\frac{At}{At+B(1-t)}\recht{log}\frac{At}{At+B(1-t)}}{t \recht{log} t} \geq \recht{min}\left\{\frac{A}{B},\frac{B}{A}\right\}.
\end{equation}
\end{lemma}

\begin{proof}
Note that the statement is trivially true if either $A= 0$ or $B = 0$, so we now suppose that neither is the case. Suppose $A \leq B$, and let $\sigma := \frac{A}{B}$. Consider the function $f\colon (0,1) \rightarrow \mathbb{R}$ given by $f(t) = \recht{log} \sigma - \recht{log}(1-(1-\sigma)t)+(1-\sigma)t\recht{log} t$.
Then $\frac{\textrm{d}^2f}{\textrm{d}t^2} = \frac{(1-\sigma)^2}{(1-(1-\sigma)t)^2}+ \frac{1-\sigma}{t} \geq 0$, so $f$ is convex; as such its supremum is the limit to one of the endpoints. Since $\lim_{t \rightarrow 1} f(t) = 0$ and $\lim_{t \rightarrow 0} f(t) = \recht{log}\sigma \leq 0$, we have $f(t) \leq 0$ for all $t \in (0,1)$. However, we can rewrite
\begin{equation}
    0 \geq f(t) = \recht{log}\frac{At}{At+B(1-t)} - \frac{At+B(1-t)}{B}\recht{log} t,
\end{equation}
which we can rewrite into
\begin{equation}
    \frac{\frac{At}{At+B(1-t)}\recht{log}\frac{At}{At+B(1-t)}}{t \recht{log} t} \geq \frac{A}{B},
\end{equation}
which is what we needed to show. If $B \leq A$, define $g(t) = \frac{\sigma}{\sigma t + 1-t}\recht{log}\frac{\sigma t}{\sigma t + 1-t}-\frac{\recht{log} t}{\sigma}$.
Then, since $\sigma \geq 1$, we have $\frac{\textrm{d}g}{\textrm{d}t} = \frac{(\sigma-1)\left[(1-t)(1-t+(t+1)\sigma)+\sigma^2 t \recht{log}\frac{1+(\sigma-1)t}{t+(\sigma-1)t}\right]}{\sigma t (1+(\sigma-1)t)^2} \geq 0$, hence $g(t) \leq \lim_{t' \rightarrow 1} g(t') = 0$. Analogous to the above we can rewrite this into
\begin{equation}
    \frac{\frac{At}{At+B(1-t)}\recht{log}\frac{At}{At+B(1-t)}}{t \recht{log} t} \geq \frac{B}{A}. \qedhere
\end{equation}  
\end{proof}

\begin{proof}[Proof of Theorem \ref{thm:wc}] Without loss of generality, we assume that $Q$ is `surjective' in the sense that for every $y \in \mathcal{B}$ there exists an $x \in \mathcal{A}$ such that $\tilde{Q}_{y|x} > 0$. For any probability measure $\mu$ on $\mathcal{P}_{\mathcal{A}}$ and $y \in \mathcal{B}$ we have
\begin{equation}
    \frac{\recht{H}(X_i|Y_i=y,P)}{\recht{H}(X_i|P)} = \frac{\mathbb{E}_{p}[\recht{H}(X_i|Y_i=y,P=p)]}{\mathbb{E}_{p}[\recht{H}(X_i|P=p)]}.
\end{equation}
Let $\mathring{\mathcal{P}}_{\mathcal{A}}$ be the interior of $\mathcal{P}_{\mathcal{A}}$, i.e. the set of those $p$ such that all $p_x > 0$. Since $\recht{H}(X_i|Y_i=y,P=p)$ and $\recht{H}(X|P=p)$ are continuous functions in $p$, we can obtain the infimum over all $\mu$ by restricting those for which $\mu(\mathring{\mathcal{P}}_{\mathcal{A}}) = 1$. Taking such a $\mu$, we can invoke Lemma \ref{lem:avg}, which now tells us that
\begin{equation}
     \frac{\recht{H}(X_i|Y_i=y,P)}{\recht{H}(X_i|P)} \geq \recht{inf}_{p \in \mathring{\mathcal{P}}_{\mathcal{A}}} \frac{\recht{H}(X_i|Y_i=y,P = p)}{\recht{H}(X_i|P = p)}. \label{eq:minp}
\end{equation}
So we are reduced from finding the optimal $\mu$ to finding the optimal $P$. Invoking Lemma \ref{lem:avg} with $\chi$ the counting measure on $\mathcal{A}$, we get
\begin{eqnarray}
\frac{\recht{H}(X_i|Y_i=y,P = p)}{\recht{H}(X_i|P = p)} &=& \frac{\sum_{x \in \mathcal{A}} \mathbb{P}(X_i=x|Y_i=y,P=p)\recht{log}\frac{1}{\mathbb{P}(X_i=x|Y_i=y,P=p)}}{\sum_{x \in \mathcal{A}} \mathbb{P}(X_i=x|P=p)\recht{log}\frac{1}{\mathbb{P}(X_i=x|P=p)}}\\
&\geq& \recht{min}_{x \in \mathcal{A}} \frac{ \mathbb{P}(X_i=x|Y_i=y,P=p)\recht{log}\frac{1}{\mathbb{P}(X_i=x|Y_i=y,P=p)}}{\mathbb{P}(X_i=x|P=p)\recht{log}\frac{1}{\mathbb{P}(X_i=x|P=p)}}\\
&=& \recht{min}_{x \in \mathcal{A}} \frac{\frac{Q_{y|x}p_x}{\sum_{x'} Q_{y|x'}p_x'}\recht{log}\frac{Q_{y|x}p_x}{\sum_{x'} Q_{y|x'}p_x'}}{p_x \recht{log} p_x}. \label{eq:minx}
\end{eqnarray}

Let $x \in \mathcal{A}$. If we apply Lemma \ref{lem:fracappr} to $t = p_x$, $A = Q_{y|x}$ and $B = \sum_{x' \neq x} Q_{y|x'}\frac{p_{x'}}{1-p_x}$, we find

\begin{eqnarray}
\frac{\frac{Q_{y|x}p_x}{\sum_{x'} Q_{y|x'}p_x'}\recht{log}\left(\frac{Q_{y|x}p_x}{\sum_{x'} Q_{y|x'}p_x'}\right)}{p_x \recht{log}(p_x)} &\geq& \recht{min}\left\{\frac{Q_{y|x}}{\sum_{x' \neq x} Q_{y|x'}\frac{p_{x'}}{1-p_x}}, \frac{\sum_{x' \neq x} Q_{y|x'}\frac{p_{x'}}{1-p_x}}{Q_{y|x}}\right\} \\
&\geq & \recht{min}\left\{\frac{Q_{y|x}}{\recht{max}_{x'}Q_{y|x'}},\frac{\recht{min}_{x'} Q_{y|x'}}{Q_{y|x}}\right\}. \label{eq:ldp}
\end{eqnarray}
By Def. \ref{def:ldp}, (\ref{eq:ldp}) is bounded from below by $\textrm{e}^{-\recht{LDP}(Q)}$. Combined with (\ref{eq:minp}) and (\ref{eq:minx}) this shows that
\begin{equation}
\recht{S}^{\recht{wc}}(Q) = \recht{inf}_{\mu} \recht{inf}_y \frac{\recht{H}(X_i | Y_i = y,P)}{\recht{H}(X_i|P)} \geq \textrm{e}^{-\recht{LDP}(Q)}. \label{eq:ldplb}
\end{equation}

On the other hand, let $y,x,x'$ be such that $\frac{Q_{y|x}}{Q_{y|x'}} = \textrm{e}^{\recht{LDP}(Q)}$. Let $\varepsilon \in (0,1)$, and let $P^{\varepsilon} \in \mathcal{P}_{\mathcal{A}}$ be of the form $p_x^{\varepsilon} = 1-\varepsilon$, $p_{x'}^{\varepsilon} = \varepsilon$, and all other probabilities equal to $0$. Let $\mu^{\varepsilon}$ be the probability measure on $\mathcal{P}_{\mathcal{A}}$ which puts all its mass on $p^{\varepsilon}$. Define $f(\varepsilon) = \recht{S}_{\mu^{\varepsilon}}(Q)$; then
\begin{equation}
f(\varepsilon) = \frac{\frac{Q_{y|x}(1-\varepsilon)}{Q_{y|x}(1-\varepsilon)+Q_{y|x'}\varepsilon}\recht{log}\frac{Q_{y|x}(1-\varepsilon)}{Q_{y|x}(1-\varepsilon)+Q_{y|x'}\varepsilon}+\frac{Q_{y|x'}\varepsilon}{Q_{y|x}(1-\varepsilon)+Q_{y|x'}\varepsilon}\recht{log}\frac{Q_{y|x'}\varepsilon}{Q_{y|x}(1-\varepsilon)+Q_{y|x'}\varepsilon}}{(1-\varepsilon)\recht{log}(1-\varepsilon)+\varepsilon \recht{log}\varepsilon}.
\end{equation}
Applying l'Hôpital's rule, we find
\begin{eqnarray}
\lim_{\varepsilon \rightarrow 0} f(\varepsilon)
&=& \lim_{\varepsilon \rightarrow 0} \frac{\left(\recht{log}\varepsilon - \recht{log}(1-\varepsilon)+\recht{log}Q_{y|x'}-\recht{log}Q_{y|x}\right)\frac{Q_{y|x'}Q_{y|x}}{(Q_{y|x'}\varepsilon+Q_{y|x}(1-\varepsilon))^2}}{\recht{log}\varepsilon - \recht{log}(1-\varepsilon)} \\
&=& \frac{Q_{y|x'}}{Q_{y|x}} = \textrm{e}^{-\recht{LDP}(Q)},
\end{eqnarray}
which shows that the lower bound in (\ref{eq:ldplb}) can indeed be obtained in the limit.
\end{proof}

From (\ref{eq:minp}), and the fact that in the proof above we use a measure $\mu^{\varepsilon}$ which puts all mass on a single point, we also get the following Corollary:

\begin{corollary}
$\recht{S}^{\recht{wc}}(Q) = \recht{inf}_{p \in \mathcal{P}_{\mathcal{A}}} \recht{inf}_{y \in \mathcal{B}} \frac{\recht{H}(X_i|Y_i = y, P = p)}{\recht{H}(X_i|P = p)}$. \qed
\end{corollary}

\section{Proofs for section \ref{sec:limit}} \label{app:limit}

In this appendix we prove the statements of section \ref{sec:limit}; it turns out these are by far the most mathematically involved. Its main result, which is Theorem \ref{thm:iypiytp}, is our main focus.

\subsection{Short proofs}

From Theorem \ref{thm:iypiytp} we can prove the statements of section \ref{sec:limit}:

\begin{proof}[Proof of Lemma \ref{lem:uas}]
For large $n$ one has $\recht{I}(\vec{Y};P) \stackrel{\infty}{\longrightarrow} \frac{a-1}{2}\recht{log} n + \recht{h}(P) + (a-1)\recht{U}^{\recht{as}}_{\mu}(Q)$ by Theorem \ref{thm:iypiytp}.1. Applying this to $\recht{id}_{\mathcal{A}}$ gives $\recht{I}(\vec{X};P) \stackrel{\infty}{\longrightarrow} \frac{a-1}{2}\recht{log} n + \recht{h}(P) + (a-1)\recht{C}_{\mu}(Q)$. Since $P \rightarrow \vec{X} \rightarrow \vec{Y}$ is a Markov chain, we have $\recht{I}(\vec{Y};P) \leq \recht{I}(\vec{X};P)$, which shows that $\recht{U}^{\recht{as}}_{\mu}(Q) \leq \recht{C}_{\mu}$.
\end{proof}

\begin{proof}[Proof of Corollary \ref{cor:limuti}]
Applying the terminology of Theorem \ref{thm:iypiytp} to the privacy protocol $\recht{id}_{\mathcal{A}}$ we find $d = a$ and $b' = a-1$. As such for general $Q$ we find
\begin{equation}
\recht{U}^{\recht{distr}}_{n,\mu}(Q) = \frac{\recht{I}(\vec{Y};P)}{\recht{I}(\vec{X};P)}
\stackrel{\infty}{\longrightarrow} \frac{\frac{d-1}{2}\recht{log} n + r_{\mu}(Q)}{\frac{a-1}{2}\recht{log} n + r_{\mu}(\recht{id}_{\mathcal{A}})} \label{eq:limutipf_1}
\stackrel{\infty}{\longrightarrow} \frac{d-1}{a-1}, 
\end{equation}
which proves (\ref{eq:limuti_1}), and 
\begin{equation}
\recht{U}^{\recht{tally}}_{n,\mu}(Q) = \frac{\recht{I}(\vec{Y};P)+\recht{I}(\vec{Y};T|P)}{\recht{I}(\vec{X};P) + \recht{I}(\vec{X};\bar{T}|P)} \stackrel{\infty}{\longrightarrow} \frac{\frac{d+b'-1}{2}\recht{log} n + r_{\mu}(Q)+s_{\mu}(Q)}{\frac{2a-2}{2}\recht{log} n + r_{\mu}(\recht{id}_{\mathcal{A}}) + s_{\mu}(\recht{id}_{\mathcal{A}})} \stackrel{\infty}{\longrightarrow} \frac{d+b'-1}{2a-2},
\end{equation}
which proves (\ref{eq:limuti_2}). For faithful $Q$ (\ref{eq:limuti_3}) follows from
\begin{equation}
\recht{U}^{\recht{digit}}_{n,\mu}(Q) = \frac{\recht{I}(\vec{Y};P)-\recht{h}(P)}{a-1} \stackrel{\infty}{\longrightarrow} \frac{\frac{a-1}{2}\recht{log} n + (a-1)\recht{U}^{\recht{as}}_{\mu}(Q)}{a-1} = \frac{1}{2}\log n + \recht{U}^{\recht{as}}_{\mu}(Q). 
\end{equation}
Furthermore, we prove (\ref{eq:limuti_4}) by writing (\ref{eq:limutipf_1}) as
\begin{equation}
\recht{U}^{\recht{distr}}_{n,\mu}(Q) \stackrel{\infty}{\longrightarrow} \frac{\frac{a-1}{2}\recht{log} n + (a-1)\recht{U}^{\recht{as}}_{\mu}(Q)+\recht{h}(P)}{\frac{a-1}{2}\recht{log} n + (a-1)\recht{C}_{\mu}+\recht{h}(P)}\stackrel{\infty}{\longrightarrow}  1 - \frac{2(\recht{C}_{\mu}-\recht{U}^{\recht{as}}_{\mu}(Q))}{\log n}. \qedhere
\end{equation}
\end{proof}

\begin{proof}[Proof of Proposition \ref{prop:fmu}]
By Theorem \ref{thm:iypiytp}.1 one has
\begin{equation}
\recht{U}^{\recht{digit}}_{\lfloor \recht{F}_{\mu}(Q) \cdot n \rfloor,\mu}(\recht{id}_{\mathcal{A}}) \stackrel{\infty}{\longrightarrow} \frac{1}{2}\log \recht{F}_{\mu}(Q) + \frac{1}{2}\log n + \recht{C}_{\mu} = \frac{1}{2}\log n + \recht{U}^{\recht{as}}_{\mu}(Q) \recht{U}^{\recht{digit}}_{n,\mu}(Q),
\end{equation}
which proves the first statement. The second statement also follows from Theorem \ref{thm:iypiytp}.1, since $\recht{U}^{\recht{digit}}_{\lfloor \varepsilon \cdot n \rfloor,\mu}(\recht{id}_{\mathcal{A}})$ grows as $\frac{a-1}{2} \recht{log} n$, whereas $\recht{U}^{\recht{digit}}_{n,\mu}(Q)$ grows as $\frac{d-1}{2}\recht{log} n$, which is strictly smaller by assumption.
\end{proof}

Sections \ref{ssec:sketch}--\ref{ssec:pfiytp} will be dedicated to proving Theorem \ref{thm:iypiytp}.

\subsection{Sketch of the proof of Theorem \ref{thm:iypiytp}.1} \label{ssec:sketch}

The first point of Theorem \ref{thm:iypiytp} is mathematically the most involved. Before we go into the proof, it is helpful to give a brief sketch of the proof. As before, let $S$ be the tallies of $\vec{Y}$ as in section \ref{ssec:ldp}.
\begin{enumerate}[wide=0pt]
\item Since $\recht{I}(\vec{Y};P) = \recht{I}(S;P) = \recht{H}(S) - \recht{H}(S|P)$, it suffices to find the asymptotic behaviour (as $n \rightarrow \infty$) of the two terms on the right hand side.
\item For a given $P = p$, the random variable $S$ is drawn from a multinomial distribution. Since the Shannon entropy of a multinomial distribution is known, this gives us a formula for the asymptotic behaviour of $\recht{H}(S|P)$ (Lemma \ref{lem:hsp}).
\item Using a multivariate version of the de Moivre--Laplace theorem, we find a continuous random variable $F_n$ on $\mathcal{P}_{\mathcal{B}}$ such that $\recht{H}(S) \approx \recht{h}(F_n) + (b-1)\log n$ for large $n$ (Lemma \ref{lem:hs}). This reduces the problem to determining the differential entropy of $F_n$.
\item We introduce a reparametrisation of (an open subset of) $\mathcal{P}_{\mathcal{B}}$ in (\ref{eq:defphi}). In this parametrisation allows us to determine $\recht{h}(F_n)$ (Lemma \ref{lem:hfn}).
\item Finally, we calculate the constant term in the case that $Q$ is faithful (Lemma \ref{lem:const})
\end{enumerate}

\subsection{Preliminaries} \label{ssec:prelim}

\underline{Coordinates on $\mathcal{P}_{\mathcal{A}}$ and $\mathcal{P}_{\mathcal{B}}$}: Throughout the rest of this appendix, and only in this appendix, we identify $\mathcal{P}_{\mathcal{A}}$ with the subset $\{(u_1,\cdots,u_{a-1}): \forall x, u_y \geq 0; \sum_x u_x \leq 1\} \subset \mathbb{R}^{a-1}$, and $\mathcal{P}_{\mathcal{B}}$ with the subset $\{(v_1,\cdots,v_{b-1}): \forall y, v_y \geq 0; \sum_y u_y \leq 1\} \subset \mathbb{R}^{b-1}$. While this notation is less `natural' than viewing elements of $\mathcal{P}_{\mathcal{A}}$ as $a$-dimensional vectors whose elements sum up to $1$, in practice this notation is more convenient for our computations. Although we regard elements $u$ of $\mathcal{P}_{\mathcal{A}}$ as $(a-1)$-dimensional vectors $(u_1,\cdots,u_{a-1})$, we will at times continue to use the notation $u_a := 1 - \sum_{x=1}^{a-1} u_x$. The analogous statement holds for elements of $\mathcal{P}_{\mathcal{B}}$. We write $\mathring{\mathcal{P}}_{\mathcal{A}}$ for the \emph{interior} of $\mathcal{P}_{\mathcal{A}}$ in $\mathbb{R}^{a-1}$, i.e. $\mathring{\mathcal{P}}_{\mathcal{A}} = \{(u_1,\cdots,u_{a-1}): \forall x, u_x > 0; \sum_x u_x < 1\}$, and we define $\mathring{\mathcal{P}}_{\mathcal{B}}$ in an analogous manner.

\noindent\underline{Privacy Mechanism}: We will make two assumptions on $Q$, which we may do without loss of generality:
\begin{enumerate}[wide=0pt]
    \item  We assume that there is no $y \leq b$ such that $\tilde{Q}_{y|x} = 0$ for all $i \leq a$; we can always ensure this by taking $\mathcal{B} = \mathcal{B}_{>0}$. This does not change the constants that play a role in Theorem \ref{thm:iypiytp}.
    \item We assume that $b > a$. We can ensure this by (repeatedly) replacing $Q$ by $Q'$ obtained as follows: assume without loss of generality that $b \in \mathcal{B}_{>0}$. Then take $Q' = (\tilde{R},\mathcal{C})$, where $\mathcal{C} = \{1,\cdots,b+1\}$, and 
    \begin{equation}
        \tilde{R}_{y|x} = \left\{\begin{array}{ll}
                                \tilde{Q}_{y|x}, & \textrm{ if $y < b$}, \\
                                \tfrac{1}{2}\tilde{Q}_{b|x}, & \textrm{ if $y \geq b$}.
        \end{array}\right.
    \end{equation}
    Then $\#\mathcal{C} = \#\mathcal{B}+1$, while $\recht{rk}(\tilde{R}) = \recht{rk}(\tilde{Q})$, $\recht{U}^{\recht{as}}_\mu(Q) = \recht{U}^{\recht{as}}_\mu(Q')$, and $\#(\mathcal{C}_{>0}/{\sim}) = \#(\mathcal{B}_{>0}/{\sim})$.  By repeating this if necessary, we can always get $b > a$ without changing the constants of Theorem \ref{thm:iypiytp}.
\end{enumerate}

In terms of the coordinates $u \in \mathcal{P}_{\mathcal{A}}$, $v \in \mathcal{P}_{\mathcal{B}}$, the map $Q_*\colon \mathcal{P}_{\mathcal{A}} \rightarrow \mathcal{P}_{\mathcal{B}}$ can be written as follows:
\begin{eqnarray}
Q_*\colon \mathbb{R}^{a-1} &\rightarrow& \mathbb{R}^{b-1} \\
u & \mapsto & \left(\tilde{Q}_{y|a} + \sum_{x=1}^{a-1} (\tilde{Q}_{y|x}-\tilde{Q}_{y|a})u_x\right)_{y < b}. \nonumber
\end{eqnarray}
Note that assumption 1 above ensures that $Q_*(\mathring{\mathcal{P}}_{\mathcal{A}}) \subset \mathring{\mathcal{P}}_{\mathcal{B}}$. We write $\mathcal{I} := Q_*(\mathcal{P}_{\mathcal{A}}) \subset \mathbb{R}^{b-1}$.

WLOG we may assume that the column space of $\tilde{Q}$ is spanned by its first $d-1$ columns and its last column. Then the map
\begin{eqnarray} 
\psi\colon \mathbb{R}^{d-1} &\rightarrow & \mathbb{R}^{b-1} \label{eq:psiidef} \\
(w_1,\cdots,w_{d-1}) &\mapsto & Q_*(w_1,\cdots,w_{d-1},0,\cdots,0) \nonumber 
\end{eqnarray}
gives a parametrisation of $\mathcal{I}$; define $\tilde{\mathcal{I}} := \psi^{-1}(\mathcal{I}) \subset \mathbb{R}^{d-1}$. Note that $\psi$ is not a linear map, but it is an affine map, i.e. there exist $\Psi \in \mathbb{R}^{(b-1)\times(d-1)}$ and $\psi_0 \in \mathbb{R}^{b-1}$ such that $\psi(w) = \Psi w+\psi_0$. Then $\Psi$ is the $(b-1)\times(d-1)$-matrix satisfying $\Psi_{y,x} = \tilde{Q}_{y|x}-\tilde{Q}_{y|a}$ for all $y < b$, $x < d$.

\noindent \underline{Multivariate de Moivre--Laplace}: If $r$ is a positive integer, $m \in \mathbb{R}^r$, and $A$ is a $r \times r$ positive definite matrix, then we write $N_{m,A}$ for the probability density function of the multivariate normal distribution with mean $m$ and covariance matrix $A$, i.e.
\begin{equation}
N_{m,A}(z) = \frac{1}{\sqrt{(2\pi)^r}\recht{det}(A)}\recht{exp}\left(-\tfrac{1}{2}z^{\recht{T}}A^{-1}z\right).
\end{equation}
If $X \sim N_{m,A}$, then the differential entropy of $X$ is given by $\recht{h}(X) = \frac{r}{2}\log(2\pi \textrm{e})+\frac{1}{2}\log(\recht{det}(C))$.

For $q$ a probability distribution on $\{1,\cdots,r\}$ with all $q_y > 0$, let $C_q$ be the $(r-1) \times (r-1)$-matrix given by $(C_q)_{y,y'} = \delta_{y=y'}q_y^{-1} + q_r^{-1}$ for all $y,y' < r$. Note that for any $z,z' \in \mathbb{R}^r$ one has 
\begin{equation}
z^{\recht{T}}C_qz' = \sum_{y=1}^{r-1}q_y z_yz_y' + q_r\left(\sum_{y=1}^{r-1} z_y\right)\left(\sum_{y=1}^{r-1} z'_y\right). \label{eq:cip}
\end{equation} 
In particular, if we substitute $z=z'$, we find that $C_Q$ is positive definite. Abusing notation slightly, for a $v = (v_1,\cdots,v_{b-1}) \in \mathbb{P}_{\mathcal{B}} \subset \mathbb{R}^{b-1}$ we will write $C_v$ for $C_{\left(v_1,\cdots,v_{b-1},v_b\right)}$. The de Moivre--Laplace theorem tells us that for large $n$, a multinomial distribution of $n$ samples can be approximated by a multivariate normal distribution:

\begin{theorem}[Multivariate de Moivre--Laplace Theorem \cite{veeh1986multivariate}] \label{thm:dmlp}
Let $L \in \mathbb{Z}^r_{\geq 0}$ be drawn from a multinomial distribution with $n$ items and probability vector $q \in \mathbb{R}^r_{> 0}$. Then for large $n$ and any $l \in \mathbb{Z}^r_{\geq 0}$ with $\sum_y l_y = n$ we have $\mathbb{P}(L = l) \approx N_{q,n^{-1}C_q^{-1}}(n^{-1}l_1,\cdots,n^{-1}l_{r-1})$.
\end{theorem}

It is convenient to have the following expression for $\recht{det}(C_q)$.

\begin{lemma} \label{lem:cdet}
Let $q$ be a probability distribution on $\{1,\cdots,r\}$. Then $\recht{det}(C_q) = \prod_{y=1}^r q_y^{-1}$.
\end{lemma}

\begin{proof}
Let $H = \recht{diag}(q_1^{-1},\cdots,q_{r-1}^{-1})$, and let $e = (1,\cdots,1)^{\recht{T}} \in \mathbb{R}^{r-1}$. Then $C_q = H + p_r^{-1}ee^{\recht{T}}$, so by the matrix determinant lemma \cite{ding2007eigenvalues} we find
\begin{equation}
\recht{det}(C_q) = (1+q_r^{-1}e^TH^{-1}e)\recht{det}(H) = \left(1+ \frac{\sum_{y=1}^{r-1} q_i}{q_r}\right)\cdot \prod_{y=1}^{r-1} q_y^{-1} = \prod_{y=1}^{r} q_y^{-1}. \qedhere
\end{equation}
\end{proof}

\noindent\underline{The random variables $M$ and $\tilde{M}$}: Let $M = Q_*P$; this is a continuous random variable on $\tilde{\mathcal{I}}$. Let $\tilde{M} = \psi^{-1}M \in \tilde{\mathcal{I}}$;
let $\Gamma$ be its probability density function, i.e. for all $U \subset \tilde{\mathcal{I}}$ we have
\begin{equation}
\mathbb{P}(\tilde{M} \in U) = \int_{w \in U} \Gamma(w)\textrm{d}w = \int_{u \in Q_*^{-1}\psi(U)} \Delta(u)\textrm{d}u.
\end{equation}

\subsection{Reduction to continuous variables} \label{ssec:cont}

The first step in proving Theorem \ref{thm:iypiytp}.1 is to reduce the calculation of $\recht{I}(\vec{Y};P)$ to the calculation of the differential entropy of a continuous random variable on $\mathbb{R}^{b-1}$. Let $S$ be the tallies of $\vec{Y}$ as in section 
\ref{ssec:ldp}; then $\recht{I}(\vec{Y};P) = \recht{I}(S;P)  = \recht{H}(S) - \recht{H}(S|P)$. The following Lemma states the asymptotic behaviour of the second term. Note that in the formula below, we regard $M \in \mathbb{R}^{b-1}$ as a $(b-1)$-dimensional vector, but we nevertheless write $M_b = 1-\sum_{y=1}^{b-1} M_y$ as in section \ref{ssec:prelim}.

\begin{lemma} \label{lem:hsp}
One has $\recht{H}(S|P) \stackrel{\infty}{\longrightarrow} \frac{b-1}{2}\recht{log}(2\pi\textrm{\emph{e}}n)+\frac{1}{2}\sum_{y=1}^b \mathbb{E}_{M} \recht{log}(M_y)$.
\end{lemma}

\begin{proof}
For a given $P = p$ the tally vector $S$ is multinomially distributed with probability vector $m = Q_*p$. As such we find, following \cite{cichon2012bernoulli}, that $\recht{H}(S|P = p) \stackrel{\infty}{\longrightarrow} \frac{b-1}{2} \recht{log}(2\pi\textrm{e}n) + \frac{1}{2} \sum_{y=1}^b \recht{log}(m_y)$. From this we deduce
\begin{equation}
\recht{H}(S|P) = \mathbb{E}_{p} \recht{H}(S|P = p) \stackrel{\infty}{\longrightarrow}  \frac{b-1}{2} \recht{log}(2\pi \textrm{e}n) + \frac{1}{2} \sum_{y=1}^b \mathbb{E}_{M}\recht{log}(M_y). \qedhere
\end{equation}
\end{proof}

Unfortunately, the description of $\recht{H}(S)$ is more complicated. We will describe its asymptotic behaviour in terms of a suitable continuous random variable. Let $F_n$ be the continuous random varirable on $\mathbb{R}^{b-1}$ taken from a multivariate normal distribution with mean $M$ and covariance matrix $n^{-1}C_M^{-1}$, where $C_M$ is as in section \ref{ssec:prelim}. Let $f_n$ be the probability density function of $F_n$; one has
\begin{equation} \label{eq:fndef}
f_n(v) = \mathbb{E}_M\left[ N_{M,n^{-1}C_M^{-1}}(v)\right].
\end{equation}

\begin{lemma} \label{lem:hs}
For large $n$ one has $\recht{H}(S) \stackrel{\infty}{\longrightarrow} \recht{h}(F_n) + (b-1)\log n$.
\end{lemma}

\begin{proof}
We know that $S$ is multinomially distributed with $n$ items and probability vector $M$. Using Theorem \ref{thm:dmlp}, we find that for a $s \in \mathbb{Z}_{\geq 0}^b$ with $\sum_y s_y = n$ one has
\begin{equation}
\mathbb{P}(S = s) \approx \mathbb{E}_M\left[N_{M,n^{-1}C_M^{-1}}(n^{-1}s_1,\cdots,n^{-1}s_{b-1}) \right]= f_n(n^{-1}s_1,\cdots,n^{-1}s_{b-1}).
\end{equation}
This shows that for large $n$, the discrete random variable $(n^{-1}S_1,\cdots,n^{-1}S_{b-1}) \in \mathcal{P}_{\mathcal{B}}$ approaches the discretisation of $F_n$ around the lattices $(n^{-1}\mathbb{Z})^{b-1}$. Since each lattice point represents a volume of $n^{1-b}$, this proves the Lemma.
\end{proof}

\subsection{Reparametrisation} \label{ssec:repar}

By Lemma \ref{lem:hs} it suffices to determine $\recht{h}(F_n) = -\int_{v \in \mathbb{R}^{b-1}} f_n(b)\log(f_n(b))\textrm{d}b$, so we want to determine the behaviour of $f_n$ as $n \rightarrow \infty$. It is convenient to introduce a reparametrisation of (part of) the space $\mathbb{R}^{b-1}$. For every $m \in \mathcal{I}$ we choose a matrix $\Omega_m$ in $\mathbb{R}^{(b-1) \times (b-d)}$ satisfying the following properties:
\begin{enumerate}[wide=0pt]
\item The map $\mathcal{I} \rightarrow \mathbb{R}^{(b-1)\times(b-d)}$ given by $m \mapsto \Omega_m$ is smooth.
\item Let $\Psi$ be as in section \ref{ssec:prelim}. Then $\Psi^{\recht{T}}C_m\Omega_m = 0_{(d-1)\times(b-d)}$ for each $m$.
\item For each $m$, let $A_m$ be the $(b-1)\times(b-1)$-matrix given in block form as $(
\Psi \ \ \Omega_m)$. Then $\recht{det}(A_m) = 1$.
\end{enumerate}
Note that we can always find such $\Omega_m$: The column vectors of $\Psi$ are linearly independent, and they span a subspace $U$ of $\mathbb{R}^{b-1}$. Since $C_m$ is positive definite, it defines an inner product on $\mathbb{R}^{b-1}$. Let $\omega_{1,m}, \cdots, \omega_{b-d,m}$ be any basis of the orthogonal complement of $U_m$ with respect to this inner product; then the matrix $\Omega_m$ whose columns are the $\omega_{j,m}$ satisfies point 2 above. By rescaling them if necessary, it also satisfies point 3 (Here we use the assumption $a < b$, since this ensures that there is at least one $\omega_{j,m}$). Since $U_m$ and $C_m$ vary smoothly with $m$, we can do this for all $m$ in such a way that point 1 is satisfied as well.

It will be convenient to work with the two matrices $J_m$ and $K_m$ introduced in the Lemma below.

\begin{lemma} \label{lem:detprod}
Let $m \in \mathcal{I}$. Define $J_m := \Psi^{\recht{T}}C_m \Psi$ and $K_m := \Omega_m^{\recht{T}}C_m \Omega_m$. Then $\recht{det}(C_m) = \recht{det}(J_m) \cdot \recht{det}(K_m)$.
\end{lemma}

\begin{proof}
Let $A_m$ be as in point 3 above. The columns of $J_m$ are orthogonal to the columns of $K_m$ with respect to the inner product defined by $C_m$, hence one has 
\begin{eqnarray}
A_m^{\recht{T}}C_m A_m = \left(\begin{array}{cc} J_m & 0_{(d-1)\times(b-d)} \\0_{(b-d)\times(d-1)} & K_m \end{array}\right).
\end{eqnarray}
Since $\recht{det}(A_m) = 1$, it follows that $\recht{det}(C_m) = \recht{det}(A_m^{\recht{T}}C_m A_m) = \recht{det}(J_m)\recht{det}(K_m)$.
\end{proof}

Define the map
\begin{eqnarray} \label{eq:defphi}
\varphi\colon\tilde{\mathcal{I}} \times \mathbb{R}^{b-d} &\rightarrow& \mathbb{R}^{b-1} \\
(w,z) &\mapsto& \psi(w) + \Omega_{\psi(w)}z = A_{\psi(w)}\binom{w}{z} + \psi_0. \nonumber
\end{eqnarray}
For $w \in \tilde{\mathcal{I}}$, one has $\mathrm{D}\varphi(w,0) = A_{\psi(w)}$. Hence $\varphi$ is injective on an open neighbourhood $\tilde{\mathcal{U}} \subset \tilde{\mathcal{I}} \times \mathbb{R}^{b-d}$ of $\tilde{\mathcal{I}} \times \{0\}$. Let $\mathcal{U} := \varphi(\tilde{\mathcal{U}})$; then $\varphi$ gives a reparametrisation of $\mathcal{U}$. In the next lemma, we use this to find an expression for $\recht{h}(F_n)$.

\begin{lemma} \label{lem:hfn}
One has $\recht{h}(F_n) \stackrel{\infty}{\longrightarrow} \recht{h}(\tilde{M}) + \frac{b-d}{2}\log(2\pi\textrm{e}n^{-1}) -\frac{1}{2} \mathbb{E}_{M}[\log \det K_M]$.
\end{lemma}

\begin{proof}
Define the continuous random variable $G_n = (G^1_n,G^2_n)$ on $\mathbb{R}^{d-1}\times\mathbb{R}^{b-d}$ by first drawing $\tilde{M}$ according to $\Gamma$, and then taking $G^1_n \sim N_{\tilde{M},n^{-1}J_{M}^{-1}}$ and $G^2_n \sim N_{0,n^{-1}K_{M}^{-1}}$. Then $A_M\binom{G^1_n}{G^2_n}+\psi_0$ follows a normal distribution with mean $A_M\binom{\tilde{M}}{0}+\psi_0 = M$, and covariance matrix $n^{-1}C_M^{-1}$; hence $A_M\binom{G^1_n}{G^2_n}+\psi_0 = F_n$. For large $n$, we see that $G^1_n$ converges in probability to $\tilde{M}$, hence for $G_n$ in $\tilde{\mathcal{U}}$ and for large $n$ we have that $F_n$ behaves as $A_{\psi(G^1_n)}\binom{G^1_n}{G^2_n}+\psi_0 = \varphi(G_n)$. As $n \rightarrow \infty$  the probability mass of $F_n$ concentrates near $\mathcal{I}$ (and that of $G_n$ near $\tilde{\mathcal{I}}$). Furthermore, for $w \in \tilde{\mathcal{I}}$ we have $|\recht{D}\varphi(w,0)| = |A_{\psi(w)}| = 1$. It follows that $\recht{h}(F_n) \stackrel{\infty}{\longrightarrow} \recht{h}(G_n)$. Again using that $G^1_n \approx \tilde{M}$, we find
\begin{equation} \label{eq:hfn}
\recht{h}(F_n) \stackrel{\infty}{\longrightarrow} \recht{h}(G^1_n,G^2_n) \stackrel{\infty}{\longrightarrow} \recht{h}(\tilde{M}) + \recht{h}(G^2_n|\tilde{M}) = \recht{h}(\tilde{M}) + \frac{b-d}{2}\log(2\pi\textrm{e}n^{-1}) -\frac{1}{2} \mathbb{E}_{M}[\log \det K_M]. \qedhere
\end{equation}
\end{proof}

The final ingredient is to relate the constant term in $\recht{H}(S)$ with $\recht{U}^{\recht{as}}_{\mu}(Q)$ in the case that $Q$ is faithful. We postpone this to the next section, and give the proof of our main result here.

\begin{proof}[Proof of Theorem \ref{thm:iypiytp}.1]
We know that $\recht{I}(\vec{Y};P) = \recht{H}(S) - \recht{H}(S|P)$. Combining the Lemmas \ref{lem:hsp}, \ref{lem:hs}, and \ref{lem:hfn}, we see that $\recht{I}(\vec{Y};P)$ converges to
\begin{equation}
\frac{d-1}{2}\recht{log}(2\pi\textrm{e}n)+\recht{h}(\tilde{M})-\frac{1}{2}\mathbb{E}_M\left[ \sum_{y=1}^b \recht{log}(M_y)+\log \det K_M\right].
\end{equation}
If $Q$ is faithful, then we apply Lemma \ref{lem:const} to find the desired formula for $r_{\mu}(Q)$.
\end{proof}

\subsection{The constant term for faithful $Q$}

The last ingredient for Theorem \ref{thm:iypiytp}.1 is Lemma \ref{lem:const}.

\begin{lemma} \label{lem:uasexp}
Let $Q$ be faithful. Then $\recht{U}^{\recht{as}}_\mu(Q) = \frac{1}{2}\log(2\pi\textrm{\emph{e}})+\frac{1}{2a-2}\mathbb{E}_M\log\recht{det}J_M$.
\end{lemma}

\begin{proof}
In light of Def. \ref{def:uas}, it suffices to show that for all $p \in \mathring{\mathcal{P}}_{\mathcal{A}}$ one has $\recht{det}(J_{Q_*p}) = \recht{det}(\tilde{Q}^{\recht{T}}D_p\tilde{Q})$. From the description of $\Psi$ in section \ref{ssec:prelim}, we find that for $x,x' < a$ we have $(J_{Q_*p})_{x,x'} = e_x^{\recht{T}}\Psi^{\recht{T}}C_{Q_*p}\Psi e_{x'} = (\tilde{Q}_{\bullet|x}-\tilde{Q}_{\bullet|a})^{\recht{T}}C_{Q_*p}(\tilde{Q}_{\bullet|x'}-\tilde{Q}_{\bullet|a})$. Applying (\ref{eq:cip}) to this, we find 
\begin{equation}
(J_{Q_*p})_{x,x'} = \sum_{y=1}^{b} \frac{(\tilde{Q}_{y|x}-\tilde{Q}_{y|a})(\tilde{Q}_{y|x'}-\tilde{Q}_{y|a})}{(Q_*p)_y}.
\end{equation}
It follows that $J_{Q_*p} = E^{\recht{T}}\tilde{Q}^{\recht{T}}D_p\tilde{Q}E$, where $E$ is the $(a-1)\times a$-matrix given in block form as
\begin{equation}
E = \left(\begin{array}{c} \recht{id}_{a-1} \\ (-1)_{1 \times (a-1)} \end{array} \right).
\end{equation}
Write $\bar{p}$ for the $a$-dimensional vector $(p_1,\cdots,p_a)$. Then for any $v \in \mathbb{R}^a$, one has $v^{\recht{T}} \tilde{Q}^{\recht{T}} D_{p} \tilde{Q} \bar{p} = v^{\recht{T}} \tilde{Q}^{\recht{T}} (1_{b \times 1}) = v^{{\recht{T}}} (1_{a \times 1}) = \sum_{i=1}^a v_i$.
As such, with respect to the inner product $B$ induced by the matrix $\tilde{Q}^{\recht{T}} D_p \tilde{Q}$ on $\mathbb{R}^a$, one has $B(\bar{p},\bar{p}) = 1$, and $\bar{p}$ is perpendicular to the column vectors of $E$. Let $\tilde{E}$ be the matrix written in block form as $(\bar{p} \ \ E)$; then
\begin{equation}
    \tilde{E}^{{\recht{T}}} \tilde{Q}^{\recht{T}} D_{\bar{p}} \tilde{Q} \tilde{E} = \left(\begin{array}{cc} 1 & 0_{1 \times (a-1)} \\ 0_{(a-1) \times 1} & J_{Q_*p}\end{array}\right);
\end{equation}
as such $\recht{det}(\tilde{E}^{{\recht{T}}} \tilde{Q}^{\recht{T}} D_{\bar{p}} \tilde{Q} \tilde{E}) = \recht{det}(J_{Q_*p})$. On the other hand, Lemma \ref{lem:detone} shows that $|\recht{det}(\tilde{E})| = 1$; hence $\recht{det}(\tilde{E}^{{\recht{T}}} \tilde{Q}^{\recht{T}} D_{\bar{p}} \tilde{Q} \tilde{E}) = \recht{det}( \tilde{Q}^{\recht{T}} D_{\bar{p}} \tilde{Q})$. Together this proves this Lemma.
\end{proof}

\begin{lemma} \label{lem:detone}
Let $n \in \mathbb{Z}_{\geq 0}$, let $v \in \mathbb{R}^n$. Then
\begin{equation}
    \left| \begin{array}{ccccc}
    v_1 & 1 & 0 & \cdots & 0 \\
    v_2 & 0 & 1 & \cdots & 0 \\
    \vdots & \vdots & \vdots & \ddots & \vdots \\
    v_{n-1} & 0 & 0 & \cdots & 1 \\
    v_n & -1 & -1 & \cdots & -1
    \end{array} \right| = (-1)^{n+1}\sum_{i=1}^n v_i. 
\end{equation}
\end{lemma}

\begin{proof}
By induction: for $n=1$ this is immediate. Now suppose the statement is true for $n-1$; then 
\begin{align}
    \left| \begin{array}{ccccc}
    v_1 & 1 & 0 & \cdots & 0 \\
    v_2 & 0 & 1 & \cdots & 0 \\
    \vdots & \vdots & \vdots & \ddots & \vdots \\
    v_{n-1} & 0 & 0 & \cdots & 1 \\
    v_n & -1 & -1 & \cdots & -1
    \end{array} \right| &=  v_1 \left| \begin{array}{cccc}
     0 & 1 & \cdots & 0 \\
     \vdots & \vdots & \ddots & \vdots \\
     0 & 0 & \cdots & 1 \\
    -1 & -1 & \cdots & -1
    \end{array} \right| - \left| \begin{array}{ccccc}
    v_2 &  1 & \cdots & 0 \\
    \vdots &  \vdots & \ddots & \vdots \\
    v_{n-1} &  0 & \cdots & 1 \\
    v_n &  -1 & \cdots & -1
    \end{array} \right| \\
    &= (-1)^{n+1}v_1  \left| \begin{array}{ccc}
      1 & \cdots & 0 \\
      \vdots & \ddots & \vdots \\
      0 & \cdots & 1
    \end{array} \right| + (-1)^{n+1}\sum_{i=2}^n v_i = (-1)^{n+1}\sum_{i=1}^n v_i. \qedhere
\end{align}
\end{proof}

\begin{lemma} \label{lem:const}
Suppose $Q$ is faithful. Then 
\begin{equation} \label{eq:consttp}
(a-1)\recht{U}^{\recht{as}}_{\mu}(Q) + \recht{h}(P) = \frac{a-1}{2} \recht{log}(2\pi \textrm{e}) +\recht{h}(\tilde{M})-\frac{1}{2}\mathbb{E}_M\left[ \sum_{y=1}^b \recht{log}(M_y)+\log \det K_M\right].
\end{equation}
\end{lemma}

\begin{proof}
Since $Q$ is faithful, we have $a = d$. In this case, we find that the map $\psi$ from (\ref{eq:psiidef}) is simply the map $Q_*$, and $P = \tilde{M}$. Furthermore, it follows from Lemmas \ref{lem:cdet} and \ref{lem:detprod} that $\sum_{y=1}^b \recht{log}(M_y)+\log \det K_M = \log \frac{\det K_M}{\det C_M} = - \log \det J_M$. Combining this with Lemma \ref{lem:uasexp}, we find
\begin{align}
(a-1)\recht{U}_\mu^{\recht{as}}(Q)+\recht{h}(P) &= \frac{a-1}{2}\log(2\pi\textrm{e})+\frac{1}{2}\mathbb{E}_M\log\recht{det}J_M+\recht{h}(P)\\
&= \frac{a-1}{2} \recht{log}(2\pi \textrm{e}) +\recht{h}(\tilde{M})-\frac{1}{2}\mathbb{E}_M\left[ \sum_{y=1}^b \recht{log}(M_y)+\log \det K_M\right]. \qedhere
\end{align}
\end{proof}

\subsection{Proof of Theorem \ref{thm:iypiytp}.2} \label{ssec:pfiytp}

It now remains to prove Theorem \ref{thm:iypiytp}.2. We will need one linear-algebraic Lemma. Recall that by assumption 1 in section \ref{ssec:prelim}, we have that $\mathcal{B} = \mathcal{B}_{>0}$.

\begin{lemma} \label{lem:dim}
For $y \in \mathcal{B}$, let $e_y$ be the $y$-th unit vector in $\mathbb{R}^{\mathcal{B}}$. Let $\sim$ be as in Def. \ref{def:sim}. Define the vector spaces
\begin{equation}
U = \left\{ \sum_{y \in \mathcal{B}} c_ye_y : \forall E \in \mathcal{B}/{\sim}, \ \sum_{y \in E} c_y = 0\right\} \label{eq:dim}
\end{equation}
and, for any $x \in \mathcal{A}$,
\begin{equation}
    V_x = \left\{ \sum_{y \in \mathcal{B}} c_ye_y : \substack{\forall y, \ (\tilde{Q}_{y|x} = 0 \  \Rightarrow \ c_y = 0),\\ \sum_y c_y = 0}\right\}.
\end{equation}
Then $U = \sum_x V_x$.
\end{lemma}

\begin{proof}
Let $\sum_y c_y e_y \in V_x$, and let $F$ be the set of all $y$ such that $c_y \neq 0$. By definition of $V_x$, this means that $\tilde{Q}_{y|x} > 0$ for all $y \in F$; hence $F$ is contained in one equivalence class $E_0$ in $\mathcal{B}$ under $\sim$. Since $\sum_y c_y = 0$, one has $\sum_{y \in E_0} c_y = 0$, and since all other coefficients are $0$ the same holds for all other orbits; hence $V_x \subset U$, from which $\sum_x V_x \subset U$ follows. For the converse inclusion, note that $U$ is generated by elements of the form $e_y - e_{y'}$ with $y \sim y'$. For such $y$ and $y'$, there exists a sequence $y_0,\cdots,y_k$ in $\mathcal{B}_{>0}$ and a sequence $x_1,\cdots,x_k$ in $\mathcal{A}$ such that $y_0 = y$, $y_k = y'$, and $Q_{y_{i-1}|x_i}, Q_{y_i|x_i} > 0$ for all $1 \leq i \leq k$. For such a sequence one has $e_{y_{i-1}}-e_{y_i} \in V_{x_i}$; hence $e_y-e_{y'} = \sum_i e_{y_{i-1}}-e_{y_i} \in \sum_x V_x$, which proves the converse inclusion.
\end{proof}

\begin{proof}[Proof of Theorem \ref{thm:iypiytp}.2]
We write $\recht{I}(\vec{Y};T) = \recht{I}(S;T) = \recht{H}(S) - \recht{H}(S|T)$. By Lemmas \ref{lem:hs} and \ref{lem:hfn} we know that $\recht{H}(S) \stackrel{\infty}{\longrightarrow} \frac{b+d}{2}\log(n)+C$, for some constant $C$. Hence we only need to determine the behaviour of $\recht{H}(S|T)$. Fix a vector $t \in \mathbb{Z}_{\geq 0}^{\mathcal{A}}$ such that $\sum_x t_x = n$, and assume that all $t_x$ are positive; this assumption is harmless as for large $n$ the probability that all $t_x$ are positive goes to $1$. For each $x \in \mathcal{A}$, let $W^x = n^{-1} S^x$, where $S^x$ is drawn from a multinomial distribution of $t_x$ samples, with probability vector $\tilde{Q}_{\bullet|x}$. Then $S = n\sum_x W^x$. The multivariate de Moivre -- Laplace theorem tells us that for large $t_x$, the discrete random variable $W^x$ is approximately equal to the $n^{-1}$-discretisation of a multivariate normal distribution. Contrary to what we have done in this appendix so far, it is here convenient to describe this as a singular normal distribution (see \cite[Section 2]{vanperlotenkleij2004}) with mean $m^x := \tilde{Q}_{\bullet|x}$, and a singular covariance matrix given by $\frac{t_x}{n^2}C^x$, where $C^x$ is the $(b \times b)$-matrix $C^x_{y,y'}  = \delta_{y=y'}\tilde{Q}_{y|x} - \tilde{Q}_{y|x}\tilde{Q}_{y'|x}$ for $y,y' \leq b$. Its support is the affine space given by
\begin{equation}
    A_x := m^x + \left\{ \sum_y c_y e_y : \substack{\forall y: \ (\tilde{Q}_{y|x} = 0 \  \Rightarrow \ c_y = 0),\\ \sum_y c_y = 0} \right\} \subset \mathbb{R}^b.
\end{equation}
Since $W = \sum_x W^x$, we can see $W$ as approximating a discretisation of a multivariate normal distribution with mean $\sum_x m^x$ whose support is the affine space $\sum_x A_x$. By Lemma \ref{lem:dim}, the dimension of this affine space is equal to that of the vector space $U$ of (\ref{eq:dim}), which is $b-b'$. The covariance matrix of this multivariate distribution is $n^{-2}\sum_x t_x C^x$, which has $b-b'$ positive eigenvalues. By the result of \cite[Prop.~2.1.7]{vanperlotenkleij2004} and the fact that $W$ is the discretisation of step size $n^{-1}$ of an $(b-b')$-dimensional continuous random variable, we find that the entropy of $\recht{H}(S|T) = \recht{H}(W|T)$ for large $n$ converges to\footnote{The term ``constant'' in (\ref{eq:metconst}) is generally not equal to the constant $\frac{b-b'}{2}\log(2\pi\textrm{e})$ from \cite{vanperlotenkleij2004}, since they compute the  differential entropy from an orthonormal basis, while our basis consists of elements of the form $e_y-e_{y'}$.}
\begin{eqnarray} \label{eq:metconst}
\mathbb{E}_{t}[\recht{H}(W|T=t)] &\stackrel{\infty}{\longrightarrow}& \mathbb{E}_{t}\left[(b-b')\log n +\tfrac{1}{2}\log\operatorname{det}'\left(\sum_x \frac{t_x}{n^2} C^x\right) + \textrm{constant} \right],
\end{eqnarray}
where $\operatorname{det}'$ is the product of the positive eigenvalues of the matrix. For large $n$, the discrete random variable $n^{-1}t$ converges to a discretisation of $P$, hence $\mathbb{E}_t\left[\log\operatorname{det}'\left(\sum_x \tfrac{t_x}{n^2} C^x\right)\right]$ converges to $\mathbb{E}_P\left[\log\operatorname{det}'\left(\sum_x \tfrac{P_x}{n} C^x\right)\right]$. which is equal to $(b'-b)\log n + \mathbb{E}_P\left[\log\operatorname{det}'\left(\sum_x P_x C^x\right)\right]$ as the matrix has $b-b'$ positive eigenvalues. Applying this to (\ref{eq:metconst}), we see that $\recht{H}(W|T) \stackrel{\infty}{\longrightarrow} \tfrac{b-b'}{2}\log n$ up to a constant Combining this with what we know for $\recht{H}(S)$, we find
\begin{equation}
\recht{I}(\vec{Y};T) = \recht{H}(S) - \recht{H}(S|T) \stackrel{\infty}{\longrightarrow} \frac{b+d-2}{2}\log n - \frac{b-b'}{2} \log n + \textrm{const.} = \frac{b'+d-2}{2} \log n + \textrm{const.} \qedhere
\end{equation}
\end{proof}

\section{Computational complexity \& specific protocols} \label{sec:comp}

\subsection{General computational complexity}

In this subsection, we describe how computationally involved the various metrics are, without making assumptions on the protocol or the privacy distribution. As a reminder, by $\langle x, y \rangle$ we mean that to compute the metric, we need to calculate the sum of $x$ different $y$-dimensional integrals, where the integrand does not scale with $n$ in complexity. We assume $b \geq a$ for simplicity.

\begin{lemma}
The complexity of calculating $\recht{U}^{\recht{distr}}_{n,\mu}(Q)$ and $\recht{U}^{\recht{digit}}_{n,\mu}(Q)$ is at most $\langle \mathcal{O}(n^{b-1}),a \rangle$.
\end{lemma}

\begin{proof}
One has $\recht{U}^{\recht{distr}}_{n,\mu}(Q) = \frac{\recht{I}(\vec{Y};{P})}{\recht{I}(\vec{X};{P})}$ and $\recht{U}^{\recht{digit}}_{n,\mu}(Q) = \frac{\recht{I}(\vec{Y};P)-\recht{h}(P)}{a-1}$, so the complexity of computing these is equal to that of computing $\recht{I}(\vec{Y};P)$ (or $\recht{I}(\vec{X};P)$, which is the same). If we let ${S}$ be as in section \ref{ssec:ldp}, then $\recht{I}(\vec{Y};{P}) = \recht{I}({S};{P})$. Write $r_s(p) := \mathbb{P}(S=s|P=p)$, then
\begin{equation}
\recht{I}(S;P) = \recht{H}(S)-\recht{H}(S|P) = \sum_s \Big(\mathbb{E}_P\big[r_s(P)\log r_s(P)\big] - \mathbb{E}_P[r_s(P)]\log\mathbb{E}_P[r_s(P)]\Big)
\end{equation}
where $r_s(p) = \prod_y \left(\sum_x \tilde{Q}_{y|x} p_x\right)^{s_y}$. We see that for every ${s}$ we need to perform two $a$-dimensional integrals. Since the number of ${s}$ is $\mathcal{O}(n^{b-1})$, this proves the lemma.
\end{proof}

\begin{lemma}
The complexity of calculating $\recht{U}^{\recht{tally}}_{n,\mu}(Q)$ is at most $\langle\mathcal{O}(n^{a-1}),a \rangle + \langle \mathcal{O}(n^{ab-1}),0 \rangle$.
\end{lemma}

\begin{proof}
Similar to the previous Lemma one has $\recht{U}^{\recht{tally}}_{n,\mu}(Q) = \frac{\recht{I}({S};{T})}{\recht{H}({T})}$. Analogous to the previous lemma one can show that the complexity of calculating $\recht{H}({T})$ is at most $\mathcal{O}(n^{a-1})$, so we focus on the numerator again. To calculate this, it is sufficient calculate all probabilities of the form $\mathbb{P}({S} = {s}, {T} = {t})$. Let $\mathcal{M}$ be the set
\begin{equation} \label{eq:complem}
    \mathcal{M} = \left\{M \in \mathbb{Z}_{\geq 0}^{\mathcal{B} \times \mathcal{A}} : \sum_y\sum_x M_{y|x} = n\right\}.
\end{equation}
An element of $\mathcal{M}$ `stands for' the situation that for every $(y,x)$ there are precisely $M_{y|x}$ users which have $X_i = x$ and $Y_i = y$. Then
\begin{eqnarray}
\mathbb{P}({S} = {s}, {T} = {t}) &=& \sum_{\substack{M \in \mathcal{M}: \\ \forall y, \sum_x M_{y|x} = s_y, \\ \forall x, \sum_y M_{y|x} = t_x}} \binom{n}{\vec{M}} \mathbb{E}_{{p}} \left[\prod_x \left( p_x^{\sum_y M_{y|x}} \prod_y \tilde{Q}_{y|x}^{M_{y|x}}\right) \right],\\
&=& \mathbb{E}_{p}\left[\prod_x p_x^{t_x}\right]\sum_{\substack{M \in \mathcal{M}: \\ \forall y, \sum_x M_{y|x} = s_y, \\ \forall x, \sum_y M_{y|x} = t_x}} \binom{n}{\vec{M}}  \prod_{x,y} \tilde{Q}_{y|x}^{M_{y|x}},
\end{eqnarray}
where $\vec{M}$ is the matrix $M$ considered as a $ab$-length vector, in order to plug it into a multinomial coefficient. Since $\#\mathcal{M} = \mathcal{O}(n^{ab-1})$, and there are $\mathcal{O}(n^{a-1})$ possibilities for $t$, this means that in total we can find $\recht{U}^{\recht{tally}}_{n,\mu}(Q)$ by performing $\mathcal{O}(n^{a-1})$ integrals in $a$ dimensions and calculating $\mathcal{O}(n^{ab-1})$ products.
\end{proof}

\begin{lemma}
The complexity of calculating $\recht{S}_{\mu}(Q)$ is at most $\langle \mathcal{O}(1),a \rangle$.
\end{lemma}

\begin{proof}
Since this metric does not depend on $n$,  we need only a constant (in terms of $n$) integrations, over the $a$-dimensional space $\mathcal{P}_{\mathcal{A}}$.
\end{proof}

\begin{lemma}
The complexity of calculating $\recht{U}^{\recht{as}}_{\mu}$ is at most as complex as taking an $a$-dimensional integral over the logarithm over a sum of $(a-1)!b^{a-1}$ terms.
\end{lemma}

\begin{proof}
Every coefficient of the matrix $J_{Q_*p}$ from Lemma \ref{lem:uasexp} is a sum of $b$ terms. Since this is a $(a-1) \times (a-1)$-matrix, the determinant can be written as a sum of $(a-1)!$ terms, each of which is a product of $a-1$ coefficients of the matrix. Writing out each of these products gives us the determinant as a sum of $(a-1)!b^{a-1}$ terms. To get $\recht{U}^{\recht{as}}_{\mu}$ we need to integrate over the logarithm of this.
\end{proof}

\begin{lemma} \label{lem:postcomp}
For a given $\vec{y}$, and assuming a Dirichlet prior, the posterior distribution is a `polynomial' with real exponents in the $p_x$ with at most $\mathcal{O}(n^{a-1})$ monomial terms. The coefficient of each of these is a sum of at most $\mathcal{O}(n^{(a-1)(b-1)})$ terms.
\end{lemma}

\begin{proof}
Let $\mathcal{M}$ be as in (\ref{eq:complem}); define $\mathcal{M}' = \{M \in \mathcal{M}: \forall y \sum_x M_{y|x} = s_y\}$, where ${s}$ is the tally vector of the given vector $\vec{y}$. Then for a given $\mathcal{M}'$, there are exactly $\prod_y \binom{s_y}{M_{y|\bullet}}$ vectors $\vec{x}$ such that $\sum_{y} M_{y|x} = t_x$ for all $x$. As such we can rewrite (\ref{eq:fpyp2}) as
\begin{equation}
f_{P|\vec{Y}=\vec{y}}(p) = \frac{1}{C_{\vec{y}}} \sum_t \left(\sum_{\substack{M \in \mathcal{M}':\\ \forall x, \sum_y M_{y|x} = t_x}}\left(\prod_y \binom{s_y}{M_{y|\bullet}} \right) \left( \prod_{x,y} \tilde{Q}_{y|x}^{M_{y|x}}\right)  \right)\recht{B}(\alpha+t)\Delta_{\alpha+t}(p).
\end{equation}
Since every Dirichlet distribution contributes a monomial to the posterior distribution, and for a given $t$ there are $\mathcal{O}(n^{(a-1)(b-1)})$ elements of $\mathcal{M}'$ whose columns sum to the $t_x$, the Lemma follows.
\end{proof}

\begin{remark}
As in section \ref{sec:posterior}, if the prior distribution is a Dirichlet distribution, there are various techniques to numerically evaluate the expected values in this appendix via Monte Carlo methods.
\end{remark}

\subsection{GRR} \label{sapp:grr}

\begin{proof}[Proof of Proposition \ref{prop:grrutipriv}]
Note that $\recht{H}(Y_i|X_i,P = p) = \recht{H}(Y_i|X_i) = -\sum_{x,y}\tilde{Q}_{y|x}p_x\log\tilde{Q}_{y|x}$. Since $\tilde{Q}_{y|x} = \frac{1+\delta_{y=x}\beta}{a+\beta}$, we can write this as
\begin{equation}
\recht{H}(Y_i|X_i) = -\sum_{x}\frac{1+\beta}{a+\beta}p_x \log \frac{1+\beta}{a+\beta}-\sum_{x\neq y}\frac{1}{a+\beta}p_x \log \frac{1}{a+\beta} = \log(a+\beta)-\frac{1+\beta}{a+\beta}\log(1+\beta).
\end{equation}
Now let us consider $\recht{H}(Y_i|P = p)$. Since $\mathbb{P}(Y_i = y|P = p)$ is equal to $\sum_x \tilde{Q}_{y|x}p_x = \frac{1+\beta p_y}{a+\beta}$, we find that
\begin{equation} \label{eq:hygivp}
\recht{H}(Y_i|P) = \mathbb{E}_{P}\left[-\sum_y \frac{1+\beta P_y}{a+\beta}\log\frac{1+\beta P_y}{a+\beta}\right] = \log(a+\beta)-\sum_y \mathbb{E}_P\left[\frac{1+\beta P_y}{a+\beta}\log(1+\beta P_y)\right].
\end{equation}
On the right hand side, since every $P_y$ has the same distribution, each summand is the same. It follows that
\begin{eqnarray}
\recht{S}_{\mu}(\recht{GRR}^{a,\varepsilon}) &=& 1-\frac{\recht{H}(Y_i|P)-\recht{H}(Y_i|X_i,P)}{\recht{H}(X_i|P)} \\
&=& 1-\frac{(1+\beta)\log(1+\beta)-a \cdot \mathbb{E}_P\left[\frac{1+\beta P_y}{a+\beta}\log(1+\beta P_y)\right]}{(a+\beta)\recht{H}(X_i|P)}.
\end{eqnarray}

Substituting $\beta+1 = \textrm{e}^{\varepsilon}$ now proves (\ref{eq:grruti}). With regards to the asymptotic utility, note that since $\tilde{Q}$ is a square matrix, in Def. \ref{def:uas} we have
\begin{equation}
    \recht{det}(\tilde{Q}^{\recht{T}} D_p \tilde{Q}) = (\recht{det}\tilde{Q})^2 \cdot \recht{det}(D_p).
\end{equation}

Let $c = 1_{a \times 1}$. Then $\tilde{Q} = \frac{\beta}{a+\beta}\recht{id}_a + \frac{1}{a+\beta}cc^{\recht{T}}$, so the matrix determinant lemma \cite{ding2007eigenvalues} tells us that
\begin{equation}
\recht{det}(\tilde{Q}) = \left(1+\frac{1}{a+\beta}c^{\recht{T}}\left(\frac{\beta}{a+\beta}\recht{id}_a\right)^{-1}c\right)\recht{det}\left(\frac{\beta}{a+\beta}\recht{id}_a\right) = \frac{\beta^{a-1}}{(a+\beta)^{a-1}}.
\end{equation}
We also know $\recht{det}(D_p) = (a+\beta)^a\prod_x (1+\beta p_x)^{-1}$. Hence $    \recht{det}(\tilde{Q}^\dagger D_p \tilde{Q}) = \frac{\beta^{2a-2}}{(a+\beta)^{a-2}}\prod_x (1+\beta p_x)^{-1}$. Plugging this into Def. \ref{def:uas} we find
\begin{equation}
\recht{U}^{\recht{as}}_{\mu}(\recht{GRR}^{a,\varepsilon}) = -\frac{1}{2}\log (2\pi\textrm{e}) + \log \beta-\frac{a-2}{2a-2}\log(a+\beta) - \frac{1}{2a-2} \sum_x \mathbb{E}_{{P}_x}\log(1+\beta P_x).
\end{equation}
Since the value of $\mathbb{E}_{{P}_x}\log(1+\beta P_x)$ does not depend on the choice of $x$ as the prior is symmetric, this proves the Proposition.
\end{proof}

From the formulas in Proposition \ref{prop:grrutipriv} we immediately get the following complexity result:

\begin{corollary}
The complexity of calculating $\recht{S}_{\mu}(\recht{GRR}^{a,\varepsilon})$ is at most $\langle \mathcal{O}(1),1 \rangle$, while to calculate $\recht{U}^{\recht{as}}_{\mu}(\recht{GRR}^{a,\varepsilon})$ one needs to perform only a $1$-dimensional integral. \qed
\end{corollary}

\begin{lemma} \label{lem:pygivp}
Let $\vec{y} \in \mathcal{A}^n$, and let ${p} \in \mathcal{P}_{\mathcal{A}}$. Let $\beta = \textrm{e}^{\varepsilon}-1$. Then
\begin{equation}
\mathbb{P}(\vec{Y} = \vec{y} | {P} = {p}) = \frac{1}{(\beta+a)^{n}} \prod_{x \in \mathcal{A}} (1+\beta p_x)^{s_x} = \frac{1}{(\beta+a)^{n}}\sum_{\substack{{k} \in \mathbb{Z}_{\geq 0}^{\mathcal{A}}:\\ \forall x \ k_x \leq s_x}}\beta^{\sum_x k_x} \left(\prod_x \binom{s_x}{k_x}p_x^{k_x} \right).
\end{equation}
\end{lemma}

\begin{proof}
As in the proof of Proposition \ref{prop:grrutipriv} one has $\mathbb{P}(Y_i = y|P = p) = \sum_x \tilde{Q}_{y|x}p_x = \frac{1+\beta p_y}{a+\beta}$, and the first equality in the Lemma is a direct consequence of this. Now note that $(1+\beta p_x)^{s_x} = \sum_{k_x = 0}^{s_x} \binom{s_x}{k_x}\beta^{k_x}p_x^{k_x}$; this proves the second equality.
\end{proof}

\begin{proof}[Proof of Proposition \ref{prop:grrpost}]
Let $\alpha$ be the constant vector of length $a$ with value $\tfrac{1}{2}$; then applying the first equation of Lemma \ref{lem:pygivp} gives us
\begin{equation} \label{eq:grrpost}
f_{P|\vec{Y}=\vec{y}}(p) = \frac{\Delta_{\alpha}(p)\mathbb{P}(\vec{Y}=\vec{y}|P=p)}{\mathbb{P}(\vec{Y}=\vec{y})} \propto \prod_{x \in \mathcal{A}} (1+\beta p_x)^{s_x}p_x^{-\frac{1}{2}}.
\end{equation}
To find the right normalisation constant, we use the second equation of Lemma \ref{lem:pygivp} to see that the right hand side of (\ref{eq:grrpost}) is equal to
\begin{equation}
\sum_{\substack{{k} \in \mathbb{Z}_{\geq 0}^{\mathcal{A}}:\\ \forall x \ k_x \leq s_x}}\beta^{\sum_x k_x} \left(\prod_x \binom{s_x}{k_x}p_x^{k_x-\frac{1}{2}}\right) = \sum_{\substack{{k} \in \mathbb{Z}_{\geq 0}^{\mathcal{A}}:\\ \forall x \ k_x \leq s_x}}\recht{B}(k+\tfrac{1}{2}) \beta^{\sum_x k_x} \Delta_{k+\frac{1}{2}}(p).
\end{equation}
Since Dirichlet distributions integrate to $1$, this gives us the correct normalisation constant.
\end{proof}

\begin{corollary}
For a given $\vec{y}$, and assuming the Jeffreys prior, the posterior distribution is a `polynomial' in the $p_x$ (with rational coefficients) with at most $\mathcal{O}(n^a)$ terms. \qed
\end{corollary}

In the case that $\mu = \mu_{\recht{Jef}}$, we can also give formulas that allow us to calculate $\recht{U}^{\recht{distr}}_{n,\mu_{\recht{Jef}}}(\recht{GRR}^{a,\varepsilon})$, at least for small $a$ and $n$, because we can reduce to only onedimensional integration:

\begin{proposition} \label{prop:grriyp}
Write $\alpha = (1/2,\cdots,1/2) \in \mathbb{R}^{a}$. Fix $n \in \mathbb{Z}_{\geq 1}$, and consider the sets
\begin{equation}
\mathcal{S}_n = \left\{s \in \mathbb{Z}_{\geq 0}^a : \sum_{y=1}^a w_y = n \right\}, \ \ \ \textrm{for } s \in \mathcal{S}_n: \ 
\mathcal{K}(s) = \left\{k \in \mathbb{Z}_{\geq 0}^a ; \forall y, \  k_y \leq s_y\right\}.
\end{equation}
For $s \in \mathcal{S}_n$ define
\begin{eqnarray}
    F_n(\beta,s) &=& \sum_{k \in \mathcal{K}(s)}\frac{\beta^{\sum_y k_y}}{(a+\beta)^n} \left(\prod_{y=1}^a \binom{s_y}{k_y}\right)\frac{\recht{B}(\alpha+k)}{\recht{B}(\alpha)};\\
    G(\beta) &=& a\cdot \mathbb{E}_L\left[\left(\frac{1+\beta L}{a+\beta}\right)\recht{log}\left(\frac{1+\beta L}{a+\beta}\right)\right],
\end{eqnarray}
where $L$ is drawn from a beta distribution with parameters $\left(\frac{1}{2},\frac{a-1}{2}\right)$. Then
\begin{eqnarray}
    \recht{I}(\vec{Y};\vec{P}) &=& n\cdot G(\beta)-\sum_{s \in \mathcal{S}_n}\binom{n}{s}F_n(\beta,s)\recht{log}(F_n(\beta,s)),\\
    \recht{I}(\vec{X};{P}) &=& n\cdot \recht{H}(X_i|P)-\sum_{s \in \mathcal{S}}\binom{n}{s}\frac{\recht{B}(\alpha+s)}{\recht{B}(\alpha)}\recht{log}\frac{\recht{B}(\alpha+s)}{\recht{B}(\alpha)}.
\end{eqnarray}
\end{proposition}

\begin{proof}
Since the $Y_i$ are independent given $P$ we have $\recht{I}(\vec{Y};{P}) = \recht{H}(\vec{Y}) - \recht{H}(\vec{Y}|{P}) = \recht{H}(\vec{Y}) - n\recht{H}(Y_i | {P})$. From (\ref{eq:hygivp}) we see that $\recht{H}(Y_i|P) = -G(\beta)$ (note that $L$, like any $P_x$, is from a beta distribution with parameters $(\tfrac{1}{2},\tfrac{a-1}{2})$). Furthermore, for $\vec{y} \in \mathcal{B}^n$ we have, from Lemma \ref{lem:pygivp} we have
\begin{equation}
    \mathbb{P}(\vec{Y} = \vec{y}) = \mathbb{E}_P\left[(\beta+a)^{-n}\sum_{\substack{{k} \in \mathbb{Z}_{\geq 0}^{\mathcal{A}}:\\ \forall x \ k_x \leq s_x}}\beta^{\sum_x k_x} \left(\prod_x \binom{s_x}{k_x}P_x^{k_x} \right)\right] = F_n(\beta,s).
\end{equation}
Since there are $\binom{n}{s}$ vectors $\vec{y} \in \mathcal{B}^n$ with tallies $s$, we get $\recht{H}(\vec{Y}) = - \sum_{s \in \mathcal{S}_n}\binom{n}{s} F_n(\beta,s) \recht{log} F_n(\beta,s)$. This gives us the expression for $\recht{I}(\vec{Y};P)$. Since $T$ follows a Dirichlet-multinomial distribution, we have $\mathbb{P}(\vec{X} = \vec{x}) = \frac{\recht{B}(\alpha+t)}{\recht{B}(\alpha)}$; the expression for $\recht{I}(\vec{X};P)$ follows from this.
\end{proof}

\begin{corollary}
The complexity of calculating $\recht{U}^{\recht{distr}}_{n,\mu}(\recht{GRR}^{a,\varepsilon})$ and $\recht{U}^{\recht{digit}}_{n,\mu}(\recht{GRR}^{a,\varepsilon})$ is at most $\langle \mathcal{O}(n^{a(a-1)}),0 \rangle + \langle \mathcal{O}(1),1 \rangle$. \qed
\end{corollary}

\subsection{UE} \label{sapp:ue}

\begin{proof}[Proof of Proposition \ref{prop:uepriv}]
Given $X_i = x$, we can consider $Y_i$ to be a vector of $a$ different independent Bernoulli variables. At position $x$, this Bernoulli variable has probability $\kappa$ of returning $1$, and at all other positions the probability is $\lambda$. As such we find $\recht{H}(Y_i|X_i) = (a-1)\recht{H}_{\recht{b}}(\lambda) + \recht{H}_{\recht{b}}(\kappa)$. Now let us consider $\recht{H}(Y_i|{P}={p})$. Note that for a given $y \in 2^{\mathcal{A}}$ and ${p} \in \mathcal{P}_{\mathcal{A}}$ one has
\begin{equation}
\mathbb{P}(Y_i = y | {P} = {p}) = \lambda^{\#y-1}(1-\lambda)^{a-\#y-1}\left(\lambda(1-\kappa)+(\kappa-\lambda)\sum_{x \in y} p_x\right).
\end{equation}
Since $\sum_{x \in y} P_x$ follows a beta distribution with parameters $(\tfrac{\#y}{2},\tfrac{a-\#y}{2})$, it follows that 
\begin{equation} \label{eq:hyigivp}
    \recht{H}(Y_i|{P}) = -\sum_{g=0}^a \binom{a}{g} \mathbb{E}_{R_g}\left[ R_g \log R_g\right].
\end{equation}
Since $\recht{S}_{\mu}(\recht{UE}^{a,\kappa,\lambda}) = 1-\frac{\recht{H}(Y_i|P)-\recht{H}(Y_i|X_i)}{\recht{H}(X_i|P)}$, this proves the Proposition.
\end{proof}

\begin{corollary}
The complexity of calculating $\recht{S}_{\mu}(\recht{UE}^{\kappa,\lambda})$ is at most $\langle \mathcal{O}(1),1 \rangle$. \qed
\end{corollary}

\begin{proposition}
Let $\alpha = (\tfrac{1}{2},\cdots,\tfrac{1}{2}) \in \mathbb{R}^a$, and define
\begin{align}
    &&\mathcal{S}_n &= \{s \in \mathbb{Z}^{2^{\mathcal{A}}}_{\geq 0} : \sum_y s_y = n \},\\
    &\textrm{for }s \in \mathcal{S}_n: & \mathcal{K}(s) &= \{k \in \mathbb{Z}^{2^{\mathcal{A}}}_{\geq 0} : \forall y, \ k_y \leq s_y \},\\
    &\textrm{for }r \in \mathbb{Z}_{\geq 0}: & \mathcal{M}^y(r) &= \left\{m^y \in \mathbb{Z}^{2^{\mathcal{A}}}_{\geq 0}: \forall x \notin y, m_x^y = 0; \sum_{x \in y} m_x^y = r\right\}.
\end{align}
For $s \in \mathcal{S}_n$, $k \in \mathcal{K}(s)$, $m_y \in \mathcal{M}^y(k_y)$ for all $y$, and $0 \leq \lambda \leq \kappa \leq 1$, define
\begin{eqnarray}
G(s,k,(m^y)_y,\kappa,\lambda) &=& \frac{\lambda^{\sum_y (s_y\#y-k_y)}(\kappa-\lambda)^{\sum_y k_y}\recht{B}\left(\alpha+\sum_y m^y\right)}{(1-\lambda)^{-n(a-1)+\sum_y s_y\#y}(1-\kappa)^{-n+\sum_y k_y}\recht{B}(\alpha)},\\
F(s,\kappa,\lambda) &=& \sum_{k \in \mathcal{K}(s)} \sum_{\substack{(m^y)_{y \in 2^{\mathcal{A}}}:\\ \forall y: m^y \in \mathcal{M}^y(k_y)}} \prod_y \left(\binom{s_y}{k_y}\binom{k_y}{m^y}\right) G(s,k,(m^y)_y,\kappa,\lambda).
\end{eqnarray}
Let $R_s$ be as in Proposition \ref{prop:uepriv}. Then under $\recht{UE}^{a,\kappa,\lambda}$, and with $P$ following the Jeffreys prior, we have
\begin{equation}
    \recht{I}(\vec{Y};{P}) = n\sum_{s=0}^a \binom{a}{s} \mathbb{E}_{R_s}\left[R_s \log R_s\right]-\sum_{s\in \mathcal{S}_n}  \binom{n}{s} F(s,\kappa,\lambda)\recht{log}\left(F(s,\kappa,\lambda)\right).
\end{equation}
\end{proposition}

\begin{proof}
From (\ref{eq:hyigivp}) one has $\recht{H}(\vec{Y}|P) = -\sum_{g=0}^a \binom{a}{g} \mathbb{E}_{R_g}\left[ R_g \log R_g\right]$. Now let $\vec{y} \in (2^{\mathcal{A}})^n$, and set $s_y := \#\{i: y_i = y\}$ for every $y \in 2^{\mathcal{A}}$. Then
\begin{eqnarray}
\mathbb{P}(\vec{Y} = \vec{y}) &=& \mathbb{E}_{{P}}\left[\prod_y \left(\lambda^{\#y-1}(1-\lambda)^{a-\#y-1}\left(\lambda(1-\kappa)+(\kappa-\lambda)\sum_{x \in y} P_x\right)\right)^{s_y}\right] \nonumber\\
&=& \frac{\lambda^{-n+\sum_y s_y\#y}}{(1-\lambda)^{-n(a-1)+\sum_y s_y\#y}}\mathbb{E}_{{P}}\left[\prod_y \left(\lambda(1-\kappa)+(\kappa-\lambda)\sum_{x \in y} P_x\right)^{s_y}\right]. \label{eq:uepriv1}
\end{eqnarray}
If we focus on the expected value, we see that it is equal to
\begin{eqnarray}
&&\mathbb{E}_{{P}}\left[\prod_y \left(\lambda(1-\kappa)+(\kappa-\lambda)\sum_{x \in y} P_x\right)^{s_y}\right]\\
&=& \sum_{k \in \mathcal{K}(s)} \frac{\lambda^{n-\sum_y k_y}(\kappa-\lambda)^{\sum_y k_y}}{(1-\kappa)^{-n+\sum_y k_y}} \left(\prod_y \binom{s_y}{k_y}\right)\mathbb{E}_{{P}}\left[\prod_y\left(\sum_{x \in y} P_x\right)^{k_y}\right]. \label{eq:uepriv2}
\end{eqnarray}
Focusing again on the expected value, we have
\begin{eqnarray}
\mathbb{E}_{{P}}\left[\prod_y\left(\sum_{x \in y} P_x\right)^{k_y}\right] &=& \sum_{\substack{(m^y)_{y \in 2^{\mathcal{A}}}:\\ \forall y: m^y \in \mathcal{M}^y(k_y)}} \prod_y\binom{k_y}{m^y} \mathbb{E}_{{P}}\left[ \prod_x P_x^{\sum_y m^y_x}\right] \nonumber\\
&=& \sum_{\substack{(m^y)_{y \in 2^{\mathcal{A}}}:\\ \forall y: m^y \in \mathcal{M}^y(k_y)}} \prod_y\binom{k_y}{m^y} \frac{\recht{B}\left(\alpha+\sum_y m^y\right)}{\recht{B}\left(\alpha\right)}. \label{eq:uepriv3}
\end{eqnarray}
Combining equations (\ref{eq:uepriv1}), (\ref{eq:uepriv2}) and (\ref{eq:uepriv3}) now gives us that $\mathbb{P}(\vec{Y} = \vec{y}) = F(w,\kappa,\lambda)$, and substituting this in $\recht{I}(\vec{Y};P) = \recht{H}(\vec{Y}) - \recht{H}(\vec{Y}|P)$ proves the Proposition.
\end{proof}

\begin{corollary}
The complexity of calculating both $\recht{U}^{\recht{distr}}_{n,\mu}(\recht{UE}^{\kappa,\lambda})$ and $\recht{U}^{\recht{digit}}_{n,\mu}(\recht{UE}^{\kappa,\lambda})$ is at most $\langle \mathcal{O}(n^{1+(a+1)2^{a-1}}),0 \rangle + \langle \mathcal{O}(1),1 \rangle$.
\end{corollary}

\begin{proof}
Note that the complexity of calculating both of these is equal to the complexity of calculating $\recht{I}(\vec{Y};P)$. One has $\#\mathcal{S}_n = \mathcal{O}(n^{2^a-1})$, for $s \in \mathcal{S}_n$ one has $\#\mathcal{K}(s) = \mathcal{O}(n^{2^a})$, and for $y \in 2^{\mathcal{A}}$ and $r \geq 0$ one has $\#\mathcal{M}^y(r) = \mathcal{O}(r^{\#y-1})$ if $\#y > 0$, and $0$ (if $r > 0$) or $1$ (if $r=0$) otherwise. Since there are $\binom{a}{g}$ choices of $y$ with $\#y = g$, the number of choices, given $k \in \mathcal{K}(s)$, for $(m^y)_y \in \prod_y \mathcal{M}^y(k_y)$ is equal to $\mathcal{O}(n^{\sum_{g=1}^a \binom{a}{g}(g-1)}) = \mathcal{O}(n^{1+(a-2)2^{a-1}})$. Combining these, we see that we have $\mathcal{O}(n^{1+(a+1)2^{a-1}})$ choices for the tuple $(s,k,(m^y)_y)$; this is the complexity of calculating $\recht{H}(\vec{Y})$. For calculating $\recht{H}(\vec{Y}|P)$, we need to compute $a+1$ onedimensional integrals.
\end{proof}

\end{document}